%% file: main.tex
\documentclass{article}
\usepackage[utf8]{inputenc}

\newcommand{\eat}[1]{}

\usepackage{amsmath, amssymb, amsfonts}

\usepackage{booktabs} 

\usepackage{graphicx}
\usepackage{amssymb}
\usepackage{amsmath}
\usepackage{graphicx}
\usepackage{latexsym}
\usepackage{epstopdf}
\usepackage{multirow}
\usepackage[vlined,ruled,commentsnumbered]{algorithm2e}
\SetEndCharOfAlgoLine{}
\usepackage{caption}
\usepackage{subcaption}
\usepackage{datetime}
\usepackage{paralist}
\setdefaultleftmargin{1em}{}{}{}{.5em}{.5em}

\usepackage{hyperref}  
\usepackage{changepage}
\usepackage{pdfpages}

\usepackage{siunitx}
\DeclareSIUnit[group-digits=true,group-separator={,}]\GB{GB}
\DeclareSIUnit[group-digits=true,group-separator={,}]\MB{MB}
\DeclareSIUnit[group-digits=true,group-separator={,}]\sec{s}
\DeclareSIUnit[group-digits=true,group-separator={,}]\msec{ms}
\DeclareSIUnit[group-digits=true,group-separator={,}]\min{min}
\usepackage{url}
\usepackage{todonotes}
\usepackage{alltt}
\usepackage{wasysym} 
\usepackage{enumitem} 
\usepackage{listings}

\usepackage{tikz}
\usepackage{tikz-qtree}

\lstdefinelanguage{scala}{
	morekeywords={abstract,case,catch,class,def,%
		do,else,extends,false,final,finally,%
		for,if,implicit,import,match,mixin,%
		new,null,object,override,package,%
		private,protected,requires,return,sealed,%
		super,this,throw,trait,true,try,%
		type,val,var,while,with,yield},
	otherkeywords={=>,<-,<\%,<:,>:,\#,@},
	sensitive=true,
	morecomment=[l]{//},
	morecomment=[n]{/*}{*/},
	morestring=[b]",
	morestring=[b]',
	morestring=[b]"""
}

\lstdefinelanguage{hypotethical-spark}{
	morekeywords={cond,var,mult,on,additive,table,foreach,PROJECT,JOIN,SUM,BY,KEY,AS,TABLE,COND,FOREACH,code},
}

\usepackage{cleveref} 

\newcommand{\mypara}[2][\smallskip]{#1\emph{\bf #2.}\xspace}
\newcommand{\mysubpara}[1]{\emph{#1.}\xspace}

\input{macros}

\newtheorem{theorem}{Theorem}
\newtheorem{lemma}[theorem]{Lemma}
\newtheorem{proposition}[theorem]{Proposition}

\newtheorem{claim}[theorem]{Claim}

\newtheorem{example}[theorem]{Example}

\newtheorem{definition}[theorem]{Definition}
\newtheorem{proof}[theorem]{Proof}

\crefname{example}{example}{examples}
\Crefname{example}{Example}{Examples}

\usepackage{paralist}

\title{Hypothetical Reasoning via Provenance Abstraction}
\author{
  Daniel Deutch\\
   \small{Tel Aviv    University}
  \and
  Yuval Moskovitch\\
   \small{Tel Aviv    University}
   \and
   Noam Rinetzky\\
   \small{Tel Aviv    University}
}
\date{}

\begin{document}

\maketitle

\begin{abstract}
    Data analytics often involves hypothetical reasoning: repeatedly
modifying the data and observing the induced effect on the computation
result of a data-centric application. Previous work has shown that
fine-grained data provenance can help make such an analysis more
efficient: instead of a costly re-execution of the underlying
application, hypothetical scenarios are applied to a pre-computed provenance expression. However,
storing provenance for complex queries and large-scale data leads to a significant overhead, which is often a barrier to the incorporation of provenance-based solutions.    

To this end, we present a framework that allows to reduce provenance size. Our approach is based on reducing the provenance granularity using user defined abstraction trees 
over the provenance variables; the granularity is based on the anticipated hypothetical scenarios.  We formalize the tradeoff between 
provenance size and supported granularity of the hypothetical reasoning, and study the complexity of the resulting optimization problem, provide efficient algorithms for tractable cases 
and heuristics for others. We experimentally study the performance of our solution for various queries and abstraction trees. Our study shows that the algorithms generally lead to substantial speedup of hypothetical reasoning, with a reasonable loss of accuracy.
\end{abstract}

\input{intro}

\input{abstraction}

\input{compression}
\input{exp}

\input{related}

\input{conc}

\paragraph*{Acknowledgements}
This research has been funded by the European Research Council (ERC) under the European Union's Horizon 2020 research and innovation programme (grant agreement No. 804302), the Israeli Ministry of Science, Technology and Space, Len Blavatnik and the Blavatnik Family foundation, Blavatnik Interdisciplinary Cyber Research Center at Tel Aviv University, and the Pazy Foundation.
The contribution of Yuval Moskovitch is part of Ph.D. thesis research conducted at Tel Aviv University.

\bibliographystyle{abbrv}
\bibliography{bibliography}

\appendix
\input{hardness-appendix}
\input{varNumExp}

\end{document}

%% file: macros.tex
\newif\ifthesis

\newif\iftr

\newcommand{\tname}[1]{\texttt{#1}}
\newcommand{\cname}[1]{\texttt{#1}}

\newcommand{\AMPLab}[1][]{\textsl{AMPLab}}

\newcommand{\ppm}{\textsl{ppm}} 

\newcommand\assignment\rho

\let\oldnl\nl
\newcommand{\nonl}{\renewcommand{\nl}{\let\nl\oldnl}}

\newcommand{\Leavesof}[1]{\ensuremath{\mathsf{L}(#1)}}
\newcommand{\Varsof}[1]{\ensuremath{\mathsf{V}(#1)}}
\newcommand{\Monsof}[1]{\ensuremath{\mathsf{M}(#1)}}
\newcommand{\numV}[1]{\ensuremath{|#1|_{_{\mathsf{V}}}}}
\newcommand{\numM}[1]{\ensuremath{|#1|_{_{\mathsf{M}}}}}
\newcommand{\subV}[2]{#1\!\!\downarrow_{#2}}
\newcommand{\numMsubV}[2]{\ensuremath{|{#1\!\!\downarrow_{#2}}|_{_{\mathsf{M}}}}} 
\newcommand{\numVsubV}[2]{\ensuremath{|{#1\!\!\downarrow_{#2}}|_{_{\mathsf{V}}}}}

\newcommand{\ID}{\mathsf{ID}}
\newcommand{\id}{\nu}
\newcommand{\bu}{n}
\newcommand{\Nat}{\mathbb{N}}
\newcommand{\range}[2][1]{\{#1..#2\}}
\newcommand{\upp}[4]{#1\langle#2,#3,#4\rangle}
\newcommand{\ignore}[1]{}
\newcommand{\realignore}[1]{}

%% file: intro.tex

\section{Introduction}\label{sec:intro}

Recent years have seen a flourish of research on {\em data provenance}. In a nutshell, provenance captures the essence of the computation performed by queries or other data-intensive applications. In particular, a prominent line of work is centered around the model of {\em provenance polynomials} (also termed how-provenance). The idea is to compute a symbolic algebraic expression whose indeterminates may correspond to tuple identifiers \cite{GKT-pods07} or tuple cells \cite{AggPaper}; the way that the indeterminates are combined through the algebraic operations of the polynomial reflects the way the corresponding tuples/cells were manipulated by the query/application.

Once generated, an important way of using provenance polynomials is to {\em valuate} the variables occurring in them, and compute the resulting value. This allows to observe the effect of hypothetical scenarios -- modifications or assertions with respect to the input data -- on the computation result. 

\begin{example}[Running example]\label{Ex:intro-example}
	Consider a telephony company's database fragment given in Figure \ref{database}. The   \tname{Cust} table contains information about the customers, including their ID, calling plan and zip code; the \tname{Calls} table contains the call durations in minutes, totaled by month for each customer; and the \tname{Plans} table contains the \emph{price per minute} (\emph{\ppm}) of every plan, where the \ppm{} may vary from month to month. The company offers several calling plans: Small business plans ($SB1$, $SB2$), enterprises plan ($E$), plans for youth ($Y1$, $Y2$) for families ($F1, F2, F3$) and for veterans ($V$), and standard plans ($A$, $B$). Each customer is subscribed to one calling plan. 
	The following query computes the company's revenues per zip code:
	%
	%
		\lstset{language=sql,basicstyle=\small\normalfont\ttfamily,escapeinside={(*}{*)}}
		\begin{lstlisting}
		SELECT Zip, SUM(Calls.Dur * Plans.Price)
		FROM Calls, Cust, Plans
		WHERE Cust.Plan = Plans.Plan
		AND Cust.ID = Calls.CID
		AND Calls.Mo = Plans.Mo
		GROUP BY Cust.Zip
		\end{lstlisting}
	
	\vspace{-3ex} 
	The query computes the revenues of the company 
	by summing the per-customer-revenue, computed by multiplying the duration of calls by the \ppm{} of 
	the customer's plan, and aggregating the result per zip code.
	An analyst working for the company may be interested in the effect of
	possible changes to the call prices on the company's revenues. For example, what if the \emph{price per minute} (\emph{\ppm}) of all plans are decreased by 20\% in March? Or what if	the \ppm{} of  
	the business calling plans are increased by $10\%$?
	To support such scenarios, we can {\em parameterize} the (multiplicative) change in price, assigning, e.g., a distinct parameter $m_i$ to capture the change in month $i$. 
	In this example we would then get as answer to the above query, instead of a single aggregate value, a symbolic provenance expression of the form 
	$654.2\cdot m_1 + \cdots + 688.8\cdot m_3 + \cdots$.
\end{example}

\input{data-fig}

\vspace{-1ex}A significant challenge preventing the practical dissemination of such solutions is their overhead:
Essentially, we replace each query result---e.g., a numerical value---by a large provenance polynomial that encompasses the computation that took place, all tuples/cells that participated in it, as well as parameters that allow to control multiple scenarios of interest. 
Empirical evidence for the practical blowup of provenance expression was shown in \cite{selP}. 
A theoretical complexity analysis of the provenance size required for hypothetical reasoning for sum queries showed that exact provisioning requires the provenance size to be exponential in the number of hypotheticals~\cite{valICDT16}.



{\em The main contribution of the present paper is a novel framework for the reduction of provenance size}. The framework is based on the notion of {\em abstraction}; the main idea is that 
instead of assigning a distinct variable per cell, we can often group variables together, forming an abstract ``meta-variable". This naturally decreases the degree of freedom upon valuation, as we are forcing the variables in the group to be uniformly assigned. In many cases, this is consistent with the required hypothetical scenarios: for instance we may assume that a discount is uniformly applied to all months of a particular quarter. In return, as we demonstrate, grouping variables together can lead to a significant reduction in the provenance size, where distinct monomials become identical.

\vspace{-2pt}\begin{example}\label{ex:reducePolySize}
	Reconsider the case where the analyst wishes to examine the effect of changing the \ppm{} of the different plans in different months. The following polynomial captures the revenues  under this hypothetical scenario for zip code 10001 (after simplifications) for the database fragment given in Figure \ref{database}.
	\vspace{-2pt}\begin{align*}
	P={} &220.8\cdot p_1\cdot m_1+ 240\cdot p_1\cdot m_3+ 127.4\cdot f_1\cdot m_1+ \\&114.45\cdot f_1\cdot m_3+ 75.9\cdot y_1\cdot m_1+ 72.5\cdot y_1\cdot m_3+ \\&42\cdot v\cdot m_1+ 24.2\cdot v\cdot m_3
	\end{align*}
	\vspace{-2pt}$p_1$ is used to control the changes in the price of plan $A$, $f_1$ for plan $F1$, $y_1$ for $Y1$, and $v$ for the veterans plan. The variables $m_i$ are used to control the plan's price per month for different months. For instance, the monomial $220.8\cdot p_1\cdot m_1$ is the result of the multiplication of the total calls duration of the customer with ID 1 in January, with the price of plan $A$ on January, parameterized by the discount-by-plan variable $p_1$ and the discount-by-month variable $m_1$.
	  
	If the analyst knows that the prices are usually changed uniformly during each quarter, a natural abstraction in this example is to group the month variables into quarterly meta-variables. I.e., replacing $m_1, \ldots, m_3$ by $q_1$, $m_4, \ldots, m_6$ by $q_2$ etc. For instance, by replacing the variables $m_1$ and $m_3$ with $q_1$, instead of the monomials $220.8\cdot p_1\cdot m_1+ 240\cdot p_1\cdot m_3$ in the above polynomial, we obtain a single monomial $460.8\cdot p_1\cdot q_1$. The resulting polynomial in this case would be smaller: 

	\vspace{1pt}
	$P= 460.8\cdot p_1\cdot q_1+ 241.85\cdot f_1\cdot q_1+  148.4\cdot y_1\cdot q_1+   66.2\cdot v\cdot q_1\,.$	 	
\end{example}

\vspace{-2pt}Our framework   consists of the following   components: 

\mysubpara{Abstraction Trees}  We develop a simple formalism that allows to capture the possible valid abstractions, namely which variables may be grouped together while allowing for the desired hypothetical scenarios. The possible abstractions are captured by trees, examples of which appear in Figures 2 and 3 for abstractions based on calling plan and quarters respectively. An abstraction is then represented by a {\em cut} in the tree separating the root from all leaves. The idea is that for every node in the chosen cut, all of its descendant leaves are replaced by a single metavariable. There may be multiple abstraction trees, e.g., each corresponding to a different attribute of the input database; in this case an abstraction corresponds to a choice of a cut in each of the trees.

\mysubpara{Optimization Problem} Formally, the problem   we study is as follows: Given a provenance polynomial and abstraction trees over (subsets of) its variables, find a choice of abstraction that reduces the provenance size, while maximizing the ``expressiveness" of the abstraction; we next explain both measures. First, the provenance size is measured by the number of monomials in the resulting provenance polynomial. The number of monomials is indeed the dominant factor in the provenance size since the size of each monomial is bounded by a typically small constant, independent from the database size (it may depend on the query or the number of hypothetical scenarios). As for the expressiveness of the abstraction, we aim at maximizing the degrees of freedom 
left for hypothetical analysis; naturally, every grouping limits the possible scenarios in the sense that it forces multiple variables to be assigned the same value. Consequently, we measure the expressiveness of the abstraction by the number of distinct variable names it defines. In sum, our goal is to reduce the number of distinct {\em monomials} in the provenance, while maximizing the number of distinct {\em variables}.

\mysubpara{Complexity and Algorithms} We characterize the complexity of the problem, focusing on multiple sub-classes. A first natural case to consider is that of a single abstraction tree (even in this case, a monomial may still consists of multiple variables, but the abstraction may apply to at most one of them); note that there may still be exponentially many cuts in the tree. 
We show that in this case 
the optimization problem is solvable in polynomial time complexity by providing an efficient algorithm.
In contrast, the problem becomes intractable (NP-hard) if 
multiple abstraction trees are used. 
Nevertheless, we provide a greedy heuristic algorithm which we experimentally show to perform well in practice. 

\mysubpara{Implementation and Experimental Results} We
have  implemented our algorithms and thoroughly tested their performance, 
as well as their sensitivity to the multiple parameters of the problem: the database size; the shape of the provenance; the size, number and structure of the abstraction trees; the required bound on the resulting provenance size, etc. To our knowledge, a benchmark for provenance abstraction does not exist, and we have thus generated one, partly based on the TPC-H benchmark. Our results indicate good scalability and robustness of our solutions as well as the effectiveness of our greedy approximation. We also show, as discussed below, that our solution significantly outperforms a previously proposed summarization technique \cite{prox}. 

\mypara[\vspace{2pt}]{Novelty} Provenance summarization was studied in multiple contexts, e.g., for probability computation \cite{DBLP:journals/pvldb/ReS08} or explanations \cite{LeeNLG17}. The main novel aspects of the present work are in the problem setting including the use of abstraction trees that both restrict and guide the summarization, and in our novel algorithms and analysis that leverage the presence of such trees. The way that we use these trees to define our optimization problem is geared towards hypothetical reasoning, where one wishes to optimize the remaining degrees of freedom for hypotheticals. The approach in  \cite{prox} is more general: the authors assume that semantic constraints over the summarization (i.e., which variables of the provenance may be grouped together) are given to the algorithm using a black-box oracle. This generality in \cite{prox}, however, leads to the lack of theoretical guarantees on the results quality and the algorithm runtime. In contrast, 
in our setting where the compression algorithms depends on and directly use the abstraction trees, we are able to provide theoretical guarantees (PTIME complexity bound as well as a correctness guarantees for the single-tree case) and bounds (NP-hardness for the general case). In addition, we 
show in \S\ref{sec:exp} that the usefulness of directly using the abstraction trees is not confined to theory, but rather also leads to substantial practical gains over the black box approach of \cite{prox}.

\mypara[\vspace{2pt}]{\bf Offline vs. Online Compression} Our solutions take as input a pre-computed provenance expression and compresses it. This is in line with previous work \cite{prox,DBLP:journals/pvldb/ReS08} and with the anticipated use case: provenance may be generated only once, upon query evaluation (which could take place using strong computing and storage capabilities \cite{edbt18}), and then sent to multiple analysts for each to perform multiple hypothetical scenarios. In such case, a crucial point for provenance size reduction to take place is before it is sent to the analysts, to minimize both the communication cost, the cost of local storage at the analysts location and perhaps most importantly, the cost of 
applying each hypothetical scenario. 
%
Still, a natural question is whether provenance compression can be applied online, 
i.e., on-the-fly alongside with query evaluation. 
This would clearly be desirable, to avoid costly provenance generation at the outset.  

In \S\ref{sec:conc}, we discuss a preliminary plan for possible 
extension of our technique to the online setting, via the idea of sampling.  
The development of such an extension raises multiple intriguing challenges for future work, including how to obtain a representative sample for this purpose and how to use such samples to estimate the provenance size for the full computation. We briefly discuss these problems in \S\ref{sec:conc}.


\mypara[\vspace{2pt}]{Paper Organization} 
\S\ref{sec:abstaction} 
defines 
our provenance model and the optimization problem we attack.
\S\ref{sec:compression}  
provides a PTIME algorithm for a tractable fragment of the problem and a heuristic for the general case. 
\S\ref{sec:exp}
describes our implementation and presents our experimental results.
\S\ref{sec:related} discusses multiple lines of related work, and 
\S\ref{sec:conc} concludes.
For space reason, the the NP-hardness proof is given in the Appendix. 


%% file: data-fig.tex
\setlength{\textfloatsep}{2pt}
\begin{figure*}
\centering
{\footnotesize
 \begin{tabular}{c}
	\begin{tabular}{|l|l|l|}
		\multicolumn{3}{c}{\tname{Cust}}\\ \hline
		\multicolumn{1}{|c|}{\cname{ID}} &
		\multicolumn{1}{c|}{\cname{Plan}} &
		\multicolumn{1}{c|}{\cname{Zip}} 
		\\ \hline
		1       &   A & 10001\\
		2       &   F1 &10001\\
		3       &   SB1&10002 \\
		4       &   Y1 &10001\\
		5       &   V &10001\\
		6       &   E &10002\\
		7       &   SB2 &10002\\
		$\dots$ & $\dots$ & $\dots$ \\
		\hline
		
	\end{tabular}
	~
	\begin{tabular}{cc}
		\multicolumn{2}{c}{\tname{Calls}}\\
		\begin{tabular}{|l|l|r|}
			\hline \multicolumn{1}{|c|}{\cname{CID}}   &
			\multicolumn{1}{c|}{\cname{Mo}} &  \multicolumn{1}{c|}{\cname{Dur}}
			\\ \hline
			1 & 1   & 522\\
			2 & 1   & 364\\
			3 & 1   & 779\\
			4 & 1   & 253\\
			5 & 1   & 168\\
			6 & 1   & 1044\\
			7 & 1   & 697\\
			$\dots$ & $\dots$   & $\dots$\\
			\hline
		\end{tabular}
		&
		\begin{tabular}{|l|l|r|}
			\hline \multicolumn{1}{|c|}{\cname{CID}}   &
			\multicolumn{1}{c|}{\cname{Mo}} &  \multicolumn{1}{c|}{\cname{Dur}}
			\\ \hline
			1 & 3   & 480\\
			2 & 3   & 327\\
			3 & 3   & 805\\
			4 & 3   & 290\\
			5 & 3   & 121\\
			6 & 3   & 1130\\
			7 & 3   & 671\\
			$\dots$ & $\dots$   & $\dots$\\
			\hline
		\end{tabular}
	\end{tabular}
	
	\\
	\\
	\begin{tabular}{cc}
		\multicolumn{2}{c}{\tname{Plans}}\\
		\begin{tabular}{|l|l|c|}
			\hline
			\multicolumn{1}{|c|}{\cname{Plan}} & \multicolumn{1}{c|}{\cname{Mo}}&\multicolumn{1}{c|}{\cname{Price}} \\
			\hline
			Plan A & 1 & 0.4\\
			Family1 (F1)& 1 & 0.35\\
			Youth1 (Y1) & 1 & 0.3\\
			Veterans (V) & 1 & 0.25\\
			Small Business1 (SB1) & 1 & 0.1\\
			Small Business2 (SB2) & 1 & 0.1\\
			Enterprise (E) & 1 &  0.05\\
			$\dots$ & $\dots$ &$\dots$\\
			\hline
		\end{tabular}
		&
		\begin{tabular}{|l|l|c|}
			\hline
			\multicolumn{1}{|c|}{\cname{Plan}} & \multicolumn{1}{c|}{\cname{Mo}}&\multicolumn{1}{c|}{\cname{Price}} \\
			\hline
			A & 3 & 0.5\\
			F1& 3 & 0.35\\
			Y1 & 3 & 0.25\\
			V & 3 & 0.2\\
			SB1 & 3 & 0.1\\
			SB2 & 3 & 0.15\\
			E & 3 &  0.05\\
			$\dots$ & $\dots$ &$\dots$\\
			\hline
		\end{tabular}
	\end{tabular}
\end{tabular}	
}
\caption{Example database}\label{database}\label{Fi:database}
\end{figure*}

%% file: abstraction.tex

\section{Problem Definition} 
\label{sec:abstaction}

This section introduces 
our provenance model and the formalisms we use  for capturing abstractions. 
We then precisely express the problems we attack and state our main results.

\subsection{Provenance Polynomials}\label{subsec:provpoly}

%

Given a set of indeterminates $X$ we use the standard notion of a polynomial over $X$ as  
a sum of monomials, where each of which is a product of indeterminates and/or rational numbers referred to as coefficients. An indeterminate may appear more than once in a monomial, in which case this number of occurrences is called its exponent.  Our model then applies to two different settings:
\begin{enumerate}
	\item When the tuples in the input database are annotated with variables and an SPJU query is executed. In this case the semiring model \cite{GKT-pods07} defines how output tuples are annotated with polynomials over this variables. The plus and times in this case are abstract semiring operation, and hypothetical scenarios correspond to assigning values to the variables. For example, these values may be Boolean, standing for hypothetical scenarios of existence/non-existence of tuples.
	
	\item We can also support queries with commutative aggregates (e.g. sum, min, max) where variables are placed/ combined with the values in certain cells for which the analyst wish to examine different valuations. In this case the plus operation in our polynomial corresponds to the aggregate function, and multiplication may correspond to its standard semantics over numbers. 
\end{enumerate}

Note that there are other, more expressive, provenance models in the literature, including support for nested aggregation and aggregation in conjunction with tuple annotations. Their support in our framework is left for future work. 
Supporting them in our framework is left for future work.

We assume that we are given a {\em multiset of such polynomials}, intuitively including all polynomials that appear in the provenance-aware result of query evaluation.

\mypara{Notations}
We denote the set of monomials in $P$ by $\Monsof{P}$.
The {\em size} of $P$, denoted by  $\numM{P}$, is defined as the number of its monomials, i.e.,  $\numM{P}=|\Monsof{P}|$. 
We denote the set of variables in $P$ by $\Varsof{P}$, i.e., $\Varsof{P}=\bigcup_{m \in \Monsof{P}}m$.
The {\em granularity} of $P$, denoted by $\numV{P}$, is defined as the number of its variables, i.e.,
$\numV{P}=|\Varsof{P}|$.
We lift these notations 
to (multi)sets of polynomials $\mathcal{P}$ in a point-wise manner, 
i.e., 
$\numM{\mathcal{P}}=\sum_{P\in\mathcal{P}}\numM{P}$,
$\Varsof{\mathcal{P}}=\bigcup_{p \in \mathcal{P}}\Monsof{P}$, and 
$\numV{\mathcal{P}}=|\Varsof{\mathcal{P}}|$.


\subsection{Abstraction 
Trees}\label{subsec:abstrees}

Our goal is to reduce the provenance polynomial size so that its number of monomials is below a given threshold, while supporting maximal granularity for hypothetical reasoning. To
this end, we allow the user to define \emph{abstraction trees} over the
variables, intuitively defining which grouping of variables, effectively forcing that the same value will be assigned to all variables in the group, ``makes sense" based on their semantics. For instance, the abstraction trees may be obtained by leveraging existing ontologies on the annotated data, in turn capturing the semantics of variables. The user may also manually construct/augment the trees based on the expected use of provenance, namely, form the trees so that variables that are expected to be assigned the same value, based on the user experience, will be located in proximity to each other in the tree.

Formally, an abstraction tree is a rooted labeled tree, where 
each node has a unique label (we thus use ``node" and ``label" interchangeably). We say that an abstraction tree $T$ is \emph{compatible} with a polynomial $P$ if
its leaves are labeled with (some of the) variables in $P$,  and
its internal nodes 
are labeled with 
\emph{meta-variables} -- these do not appear in $P$. Additionally, 
every monomial contains at most one node from $T$, i.e., $\forall m \in \Monsof{P}. |m \cap T| \leq 1 $. Different variables in a monomial of the provenance polynomial usually originate from different domains, e.g., different scenarios (as in our running example), or different relations (as in the model of \cite{GKT-pods07}). Thus, it is reasonable to assume they are abstracted using different abstraction trees.
(For simplicity, in the following we omit the distinction between variables and 
meta-variables.)

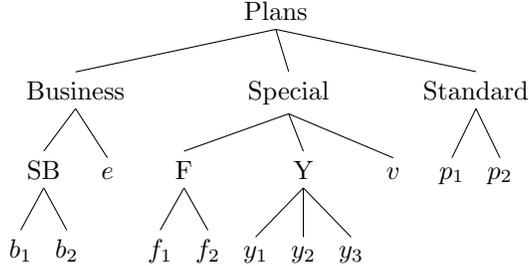
\begin{figure}[t]
	\centering
	\begin{tikzpicture}
	\Tree [.Plans
	[.Business [.SB $b_1$ $b_2$ ]  $e$ ]
	[.Special [.F $f_1$ $f_2$ ] [.Y $y_1$ $y_2$ $y_3$ ] $v$ ]
	[.Standard $p_1$ $p_2$ ]
	]
	\end{tikzpicture}
	\caption{An 
	abstraction tree of the 
	plans variables.}
	\label{fig:plansLattice}
\end{figure}

\begin{figure}[t]
	\centering
	\begin{tikzpicture}
	\tikzset{level 1+/.style={sibling distance=2\baselineskip}}
	\Tree
	[.Year [.$q_1$ $m_1$ $m_3$ ] [.$q_4$ $m_{10}$ $m_{12}$ ]]
	\node (x) at (0,-1) {$\ldots$} ;
	\node (x) at (-1.5,-2) {$\ldots$} ;
	\node (x) at (1.5,-2) {$\ldots$} ;
	\end{tikzpicture}
	\caption{An abstraction tree of the months variables.}
	\label{fig:monthLattice}
\end{figure}
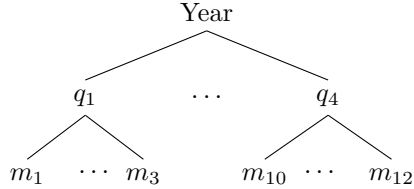

\vspace{-6pt}\begin{example}\label{ex:abstrees}
    In  Example~\ref{Ex:intro-example}, every variable  
	$m_i$ corresponds to the discount at the $i$'th month. 
	A natural abstraction over this set  is to use a set of variables per quarter $\{q_1,\ldots,q_4\}$, as shown in Example \ref{ex:reducePolySize}. Using the quarter's variables reduces the polynomials size on the one hand, but also reduces the provenance granularity; the corresponding tree is shown in Figure~\ref{fig:monthLattice}.
	
	Similarly,  
	the plans variables may be abstracted based on their 
	type, e.g., plans for small businesses $SB1$ and $SB2$, or further abstracting all Business 
	 plans, small and enterprises. Abstracting the families plans using a single variable $F$, and youth plans using the variable $Y$. We may also consider using a coarse abstraction that combines all special plans (families, youth and veterans) into a single variable. Figure \ref{fig:plansLattice} depicts the resulting abstraction tree. Note that every monomials of the polynomial 
	in Example \ref{ex:reducePolySize} contains at most one month variable and one plans variable. 
\end{example}


\subsection{Abstract Provenance Polynomials} 

Let $T$ be an abstraction 
tree and denote the set of the 
variables it contains, i.e., its nodes, by $\Varsof{T}$; the leaves, that correspond to variables occurring in the polynomial, are denoted by $\Leavesof{T}$. 
We say that $v' \leq_T v$
if $v, v'\in \Varsof{T}$ and $v'$ is a descendant of $v$ in $T$ (or $v' =
v$). 
A set of abstraction 
trees $\mathcal{T}$ is a \emph{valid abstraction forest}
if its trees are disjoint, i.e., if $T,T'\in\mathcal{T}$ and $\Varsof{T} \cap \Varsof{T'}\neq \emptyset$ then $T=T'$.
Given a set of abstraction 
trees $\mathcal{T}$, 
we denote by $\Varsof{\mathcal{T}}$  the set of all nodes of its trees,
i.e., $\Varsof{\mathcal{T}}=\bigcup_{T\in \mathcal{T}} \Varsof{T}$.
We say that $v' \leq_{\mathcal{T}} v$ if there exists an abstraction 
tree $T\in\mathcal{T}~$ such that 
$v' \leq_T v$; we omit $\mathcal{T}$ if it is clear from context. 
$\mathcal{T}$ is \emph{compatible} with a polynomial $P$ if all of its trees are; $\mathcal{T}$ is \emph{compatible} with a set of polynomials $\mathcal{P}$ if it is compatible with every polynomial in it.

We next define the notion of a \emph{valid variable set} (VVS for short), capturing a choice of abstraction using the trees.
 \begin{definition}
    Let $\mathcal{T}$ be a \emph{valid abstraction forest}.
	A \emph{valid variable set (VVS)} $S$ is a subset of $\Varsof{\mathcal{T}}$ such that
	\begin{enumerate}
		\item $\forall v\in \Leavesof{T}~ \exists v'\in S$ such that $v \leq_{\mathcal{T}} v'$ and 
		\item if $v,v'\in S$ and $v \leq_{\mathcal{T}} v'$ then $v = v'$.
	\end{enumerate}
\end{definition}


\vspace{-2pt}Technically, a valid variable set is a cut in the tree, separating the root from the leaves. 
which dictates a choice, for each leaf $v$, of a single ancestor $v'$in the tree, chosen to be its abstraction (it may be that $v=v'$, 
in which case $v$  
is not to abstracted). Intuitively, such choice means that for the subsequent hypothetical reasoning scenarios, all variables below each chosen node must be assigned the same value. 

\vspace{-2pt}\begin{example}\label{ex:validVarSet}
	Consider the abstraction tree presented in Figure \ref{fig:plansLattice}.
	The following sets 	are all valid variable sets:
	
	\noindent\hspace{-2ex}$
	\begin{array}{l@{}l}
	\begin{array}{l}
	S_1 = \{Business, Special, Standard\}\\ 
	S_2 = \{SB, e, f_1,f_2, Y, v, Standard\}\\ 
	S_3 = \{b_1, b_2, e, Special, Standard\}\\ 
	\end{array}
	&
	\begin{array}{l}
	S_4 = \{SB, e, F, Y, v, p_1, p_2\}\\ 
	S_5 = \{Plans\}\\
	\\
	\end{array}
	\end{array}
	$
\end{example}

\vspace{-2pt}Given a valid abstraction forest $\mathcal{T}$, which is compatible with 
a polynomial $P$ and a valid variable set $S$ of $\mathcal{T}$, 
we denote by $\subV{P}{S}$ the polynomial obtained from $P$ by substituting each variable
$v$ in $P$ by the (necessarily unique) variable $v'\in S$ such that $v \leq v'$, if such a variable $v'$ exists. (Recall that some variables of $P$ may not appear in the abstraction forest, in which case they stay intact. Also if $v \in S$ then $v'=v$). 
We lift the definition of $\subV{P}{S}$ to (multi)sets of polynomial $\mathcal{P}$ in a point-wise manner, i.e., $\subV{\mathcal{P}}{S}=\{\subV{P}{S} \mid P \in \mathcal{P}\}$.

Each choice of a valid variable set $S$ may entail a  ``loss" in terms of the granularity of hypothetical reasoning, in exchange for a reduction in the size of the polynomial. We quantify granularity through the number of distinct variables ($\numVsubV{P}{S}$) and size through the number of monomials ($\numMsubV{P}{S}$).

\begin{example}\label{ex:comprssionBK}
	Consider the polynomial $P$ for the revenues shown in Example \ref{ex:reducePolySize}. 
	Example \ref{ex:validVarSet} discusses 
	multiple ways to compress it, 
	each may result in
	different granularities and sizes. E.g., using $S_1$ we get that 
	$\numVsubV{P}{S_1} = 4$ and $\numMsubV{P}{S_1} = 4$, 
		whereas using $S_5$ we get that $\numVsubV{P}{S_5} = 3$ and $\numMsubV{P}{S_5} = 2$.
\end{example}

\subsection{Problems}
To formulate the problems we attack, 
we first define the notions
of \emph{precise}, \emph{adequate}, and \emph{optimal} abstractions, which, intuitively,  
state when a chosen abstraction (VVS) 
is powerful enough to reduce the size of the provenance under a given bound while maintaining sufficient granularity.

\begin{definition}[Precise, Adequate, and Optimal Abstractions]\label{def:adeopt}
Let $\mathcal{T}$ be a valid abstraction forest  compatible with a given set of polynomials $\mathcal{P}$.
Let $S$ be a valid variable set  of $\mathcal{T}$,
and $B\in \{1..\numM{\mathcal{P}}\}$ and  $K\in\{1..\numV{\mathcal{P}}\}$ 
desired bounds  on the size and granularity, respectively, of the set of polynomials 
resulting from abstracting $\mathcal{P}$ according to $S$. 
\begin{itemize}
\item
$S$ is \emph{precise} for $B$ and $K$ if $\numMsubV{\mathcal{P}}{S} = B$ and $\numVsubV{\mathcal{P}}{S}=K$. 
\item 
$S$ is \emph{adequate} for $B$ if $\numMsubV{\mathcal{P}}{S} \leq B$.
\item 
$S$ is \emph{optimal} for $B$ if it is adequate for $B$   
and 
if a valid variable set $S'$  is adequate for $B$ then
$\numVsubV{\mathcal{P}}{S'} \leq \numVsubV{\mathcal{P}}{S}$.

\end{itemize}
\end{definition}

Note that for a given set of polynomials $\mathcal{P}$, abstraction forest $\mathcal{T}$, and a bound $B$, the existence of a VVS $S$ that is adequate for $B$ is not guaranteed.

\begin{example}
	Consider the set $\mathcal{P}$ with the single polynomial $P$ from Example \ref{ex:reducePolySize} and the abstraction forest $\mathcal{T}$ with a single tree $T$ shown in Figure \ref{fig:monthLattice}. The maximal  compression of $P$ using $T$ results in a polynomial of size 4, thus, for the bound $B = 3$, there is no VVS $S$ adequate for $B$.
\end{example}

Formally, the optimization problem that we study 
is that of finding an optimal abstraction for a given bound~$B$.

\begin{definition}[Optimization problem]\label{def:optProblem}
Given a set of polynomials $\mathcal{P}$  and a valid abstraction forest $\mathcal{T}$ which is compatible with it, and a bound $B\in \{1 \ldots \numM{\mathcal{P}}\}$, find an \emph{optimal} valid variable set $S$ for $B$.
\end{definition}

To formally characterize the complexity of the optimization problem, we further need to define a corresponding decision problem. We define it as the problem of determining existence of a precise abstraction for
a given desired size  $B$ and granularity  $K$ of the reduced polynomial.
Note that as both $B$ and $K$ are polynomial in the size of the provenance,
 the optimization problem can be solved via a polynomial number of invocations of a solution to the decision problem.


\begin{definition}[Decision problem]\label{def:decisionProblem}
Given a set of polynomials $\mathcal{P}$, a valid abstraction forest $\mathcal{T}$ which is compatible with it, and desired size $B\in \{1 \ldots \numM{\mathcal{P}} \}$ and  granularity $K\in\{1 \ldots \numV{\mathcal{P}} \}$ of the set of polynomials
resulting from abstracting $\mathcal{P}$ according to $S$. 
Determine if there exists a precise valid variable set for $B$ and $K$.
\end{definition}


\subsection{Main Results} \label{subsec:main_res}
The main results of this paper can be summarized as follows:
\begin{compactenum}
\item 
We characterize in detail the complexity of the problem. We show that the problem is NP-Hard in general (\S\ref{sec:hardness-app}).
\begin{proposition}\label{prop:main}
	The decision problem is NP-hard
\end{proposition}

\item 
We study an important restricted case, where we are given a single abstraction tree. Recall that each abstraction tree allows further compression; we experimentally show that even with a single abstraction tree it is possible to reduce the polynomial size significantly. For this restricted case, we show its tractability by providing an efficient algorithm (\S\ref{sec:alg:single-tree}).  
\begin{proposition}\label{prop:ptime}
	The optimization problem is in PTIME if the abstraction forest contains exactly one tree. 
\end{proposition}


\item
We provide a simple yet effective greedy algorithm for abstracting polynomials using multiple trees 
(\S\ref{sec:greedy}).

\end{compactenum}

%% file: compression.tex
\section{Algorithms}
\label{sec:compression}
 
In this section, we show that the optimization  problem is tractable when we are given a single abstraction tree, and present a greedy algorithm for the general case.

\subsection{Single Abstraction Tree}\label{sec:alg:single-tree}

\Cref{algo:varSelection}  allows to compress a (multi)set of polynomials $\mathcal{P}$ using a single abstraction tree $T$ in an optimal manner in polynomial time.
In fact, given a bound $B=1..\numM{\mathcal{P}}$ on the desired size of the compressed provenance expression, 
the algorithm computes for each node $v \in T$ and every  $i=1..|B|$ an optimal valid variable set $S^i_v$ 
such that
(a) $S^i_v$ is comprised of metavariables coming from the subtree rooted at $v$ 
and
(b) abstracting  $\mathcal{P}$ according to $S^i_v$ results in polynomial containing $i$  monomials.
($S^i_v$ is set to $\bot$ if no such VVS exists.)

Roughly speaking, \Cref{algo:varSelection} works as follows:
It first calculates the optimal sets $S^i_v$ for the leaves $v\in\Leavesof{T}$.
It then  traverses the tree in a bottom up manner, computing the 
required optimal sets $S^i_v$ for every internal node $v$ using the optimal sets computed for its children or by setting  $S^i_v=\{v\}$. 

The key insight is that for any two nodes $v$ and $v'$
and any $i,j=1..B$, 
if neither $v  \leq v'$ nor $v'  \leq v$ holds then  $S=S^i_v \cup S^j_{v'}$ is a VVS
and that $S^i_v \cap S^j_{v'}=\emptyset$.
Furthermore, recall that we require that any monomial in $\mathcal{P}$ may contain at most one variable in $T$. 
Thus, the disjointness of $S^i_v$ and $S^j_{v'}$ ensures that 
they compress $\mathcal{P}$ by unifying disjoint sets of monomials.
Hence, $\numMsubV{\mathcal{P}}{S}$ (the size of   the provenance expression obtained by abstracting  $\mathcal{P}$ according to $S$) can be determined from the sizes of 
$\numMsubV{\mathcal{P}}{S^i_v}$ and $\numMsubV{\mathcal{P}}{S^j_{v'}}$;
namely, 
$\numMsubV{\mathcal{P}}{S} =
 \numM{\mathcal{P}} - 
 ( 
  (\numM{\mathcal{P}} - \numMsubV{\mathcal{P}}{S^i_v}) 
  + 
  (\numM{\mathcal{P}} - \numMsubV{\mathcal{P}}{S^j_{v'}})
 ) $.
For the same reasons,  we get that 
$\numVsubV{\mathcal{P}}{S} =
 \numV{\mathcal{P}} - 
 ( 
  (\numV{\mathcal{P}} - \numVsubV{\mathcal{P}}{S^i_v}) 
  + 
  (\numV{\mathcal{P}} - \numVsubV{\mathcal{P}}{S^j_{v'}})
 ) $.
Based on these observations, \Cref{algo:varSelection} uses a dynamic programming
technique to compute the optimal VVSs of internal nodes, as explained below.

\mypara[\vspace{2pt}]{Notations}
For convince, we introduce two derived properties which correspond to the \emph{difference}  in the size and granularity of the original polynomial and the compressed one: 
Given a (multi)set of polynomials $\mathcal{P}$ and valid variable set $S$, 
the induced \emph{monomials loss},
denoted by
$ML_{\mathcal{P}}(S)$, is the difference between the number of monomials in $\mathcal{P}$ and in $\subV{\mathcal{P}}{S}$, i.e., $ML_{\mathcal{P}}(S) = \numM{\mathcal{P}} - \numMsubV{\mathcal{P}}{S}$.
Similarly, the induced \emph{variable loss},
denoted by $VL_{\mathcal{P}}(S)$, is the difference between the number of
distinct variables in $\mathcal{P}$ and in $\numV{\subV{\mathcal{P}}{S}}$,
i.e.,  $VL_{\mathcal{P}}(S) = \numV{\mathcal{P}} - \numVsubV{\mathcal{P}}{S}$. 
(We omit the $\mathcal{P}$ subscript when it is clear from context.)
For instance, 
in Example \ref{ex:comprssionBK}, 
$ML(S_1)=4$ and $ML(S_5) = 6$, while $VL(S_1)=2$ and $VL(S_5)=3$.
Also note that for any VSSs $S'$ and $S''$, if  $S' \cup S''$ is also a VVS then
$ML(S' \cup S'') =ML(S') \cup ML(S'')$ and 
$VL(S' \cup S'') =VL(S') \cup VL(S'')$.

\SetInd{1.3ex}{1.3ex}
\begin{algorithm}
	\SetKwInOut{Input}{input}\SetKwInOut{Output}{output}
	\LinesNumbered
	\Input{A set of polynomials $\mathcal{P}$, a variable abstraction tree $T$ and bound $B$}
	\Output{Optimal valid variables set $S$} \BlankLine
	\SetKwFunction{algo}{optimalVVS}\SetKwFunction{proc}{computeArray}
	\SetKwFunction{ML}{ML}\SetKwFunction{VL}{VL}
	\SetKwProg{myalg}{Algorithm}{}{}
	\nonl \myalg{\algo{$\mathcal{P}, T, B$}}{
		$k \gets \numM{\mathcal{P}}-B$\;\label{line:k}
		\ForEach{leaf $l$ in $T$}
		{\label{line:forEachLeafStart}
			$A_{l}[0] \gets 0$\;
			\For{$i=1$ to $k$}
			{
				$A_{l}[i] \gets \bot$\;\label{line:forEachLeafEnd}
			}
		}
		\nonl Traverse $T$ in a bottom-up fashion\;
		\ForEach{(non-leaf) node $v$ in $T$}
		{\label{line:forEachNodeStart}
			\nonl Let $v_1,\ldots,v_m$ be the immediate children of $v$\;
			$A_{v} \gets$\proc($A{v_1},\ldots,A{v_m})$\;\label{line:computeArray}
			\If{\ML($v$) $< k$}{
				$A_{v}[$\ML$(v)] \gets \min(A_{v}[$\ML$(v)], $\VL($v$))\label{line:computeNodeML1}
			}
			\Else{
				$A_{v}[k] \gets \min(A_{v}[k], $\VL$(v))$\label{line:computeNodeML2}
			}
		}
		\Return $A_{r}[k]$ where $r$ is the root of $T$\; \label{line:returnVVS}
	}
	\setcounter{AlgoLine}{0}
	\SetKwProg{myproc}{Procedure}{}{}
	\nonl \myproc{\proc{$A{v_1},\ldots,A{v_m}$}}{
		\For{$j=0$ to $k$}
		{\label{line:initStart}
			$\tau[1,j]\gets A_1[j]$\; \label{line:init1}
			\For{$i=2$ to $m$}
			{
				$\tau[i,j]\gets \bot$\; \label{line:init2}
			}
		}
		\For{$i=2$ to $m$}
		{\label{line:forChildStart}
			\For{$j=0$ to $k$}
			{\label{line:forEntryStart}
				\For{$s=0$ to $j$}
				{
					\If{$\tau[i-1,s] \neq \bot$ and $A_i[j-s] \neq \bot$}
					{
						\If{$\tau[i,j] = \bot$}
						{
							$\tau[i,j] \gets \tau[i-1,s]+A_i[j-s]$\;
						}
						\ElseIf{$\tau[i-1,s]+A_i[j-s]<\tau[i,j]$}
						{
							$\tau[i,j] \gets \tau[i-1,s]+A_i[j-s]$\;
						}
					}
					
				}
			}\label{line:forChildEnd}
		}
		\For{$j=0$ to $k$}
		{\label{line:constructRes}
			$A_v[j] \gets T[m,j]$\;
		}
		\Return $A_v$\;  \label{line:return}
		
	}
	\caption{Optimal Valid Variables Selection}\label{algo:varSelection}
\end{algorithm} 

\mypara[] {\Cref{algo:varSelection}} 
The algorithm 
keeps for every node $v\in T$ an  array $A_v$  which encodes an optimal VVS for every possible monomial loss ranging from $0$ to $k=\numM{\mathcal{P}}-B$. 
Every array has $k+1$ entries. 
The $i$'th  entry, for $i=0..k$, records the set $S^i_v$ (see above) and is set to $\bot$ if no such VVS exists. 
If $S^{k}_v$ does not exists then $A_v[k]$   
records the minimal  $VL(S_v)$ such that $ML(S_v) \geq k$   (if no such $S_v$ exists then $A_v[k] =\bot$).

The algorithm computes $A_v$ is a bottom-up fashion, and keeps pointers to the tree nodes used to obtain the minimal variable loss  for each entry. The returned value is the VVS encoded by the $\numM{\mathcal{P}}-B$ entry of the root array.

More specifically, the algorithm initializes the  $A_l[0]$ 
for every leaf node $l \in \Leavesof{T}$ to zero and sets all other entries to $\bot$.
(lines \ref{line:forEachLeafStart}--\ref{line:forEachLeafEnd}). 
Note that indeed, a VVS which does not contain any metavariable cannot compress any polynomial.
 
Now, let $v$ be a node with children $v_1,\ldots,v_m$ in the abstraction
tree. Given the arrays $A_{v_j}$ for $1 \leq j \leq m$, $A_{v}$ can be computed. 
(Recall that a VVS $S_v$ can be either
a union of valid variable sets $S_{v_j}$ for $1 \leq j \leq m$ or
$S_v = \{v\}$.) If $S_v = \bigcup_{j=1}^m S_{v_j}$ then $VL(S_v) =
\sum_{j=1}^m VL(S_{v_j})$ and $ML(S_v) = \sum_{j=1}^m ML(S_{v_j})$.
In the case where $S_v = \{v\}$, $VL(S_v)$ equals the
number of leaves in the sub-tree rooted at $v$ minus 1, and $ML(S_v)$
can be easily computed by abstracting $\mathcal{P}$ using the VVS $\{v\}$.

For each node $v$, the algorithm uses the procedure \texttt{compute\-Array()} to compute the optimal possible valid variables sets for each $0 \leq i \leq k$ using $v$'s children VVS (lines \ref{line:forEachNodeStart}--\ref{line:computeArray}). Then using the procedures \texttt{ML($v$)} and \texttt{VL($v$)} that computes $ML(S_v)$ and $VL(S_v)$ for the case where $S_v = \{v\}$ respectively, and updates $A_v$ accordingly (lines \ref{line:computeNodeML1}--\ref{line:computeNodeML2}).
\texttt{compute\-Array()} determines the $i$'th entry in $A_v$ using  dynamic programming, where the outer loop (line \ref{line:forChildStart}) computes the the optimal VVSs considering only VVSs containing meta variables from the
subtrees rooted at $A_{v_1},\ldots,A_{v_i}$.
Finally, the optimal VVS $S$ with $ML(S)
\geq \numM{\mathcal{P}}-B$, is the VVS used to obtain the value in $A_{r}[k]$ where $r$
is the root of $T$ and $k = \numM{\mathcal{P}}-B$ (line \ref{line:returnVVS}). 
Constructing the valid variable set is then done by pointer chasing.

\begin{example}\label{ex:optAlgoRun}
The following polynomials are resulting from the query of our running example (Example \ref{Ex:intro-example}), the database fragment given in Figure \ref{database}, and parameterization based on the leaves of the abstraction trees in Figures \ref{fig:plansLattice} and \ref{fig:monthLattice}. The variables $p_1$, $f_1$, $y_1$, $v$ $b_1$, $b_2$ and $e$ are used to parameterize the plans prices based on the plan's type, and $m_1$ and $m_3$ are used control the planes prices per month for different months.
	\begin{align*}
		P_1={} &220.8\cdot p_1\cdot m_1+ 240\cdot p_1\cdot m_3+ 127.4\cdot f_1\cdot m_1+ \\& 114.45\cdot f_1\cdot m_3+  75.9\cdot y_1\cdot m_1+ 72.5\cdot y_1\cdot m_3+ \\&42\cdot v\cdot m_1+ 24.2\cdot v\cdot m_3 \\
		P_2={} & 77.9\cdot b_1\cdot m_1+ 80.5\cdot b_1\cdot m_3+ 52.2\cdot e\cdot m_1+ \\&56.5\cdot e\cdot m_3+ 69.7\cdot b_2\cdot m_1+ 100.65\cdot b_2\cdot m_3
		\end{align*}
		Given the polynomials set $\{P_1, P_2\}$, the plans abstraction tree\footnote{We assume that all the leaves of the given tree appears in the polynomials. Otherwise, 
		we ``clean'' the tree by removing redundant nodes.} shown in Figure \ref{fig:plansLattice} and the bound $B=9$, The algorithm first computes 	$k = \numM{\mathcal{P}}-B = \numM{P_1} + \numM{P_2} - B = 8 + 6 - 9 = 5$. After initializing the leaves arrays, the algorithm traverses the tree in a bottom-up fashion and computes the array for the inner nodes. In this example, in the array of $SB$ $A_{SB}[2]= 1$ because the abstraction that use the variable $SB$ instead of $b_1$ and $b_2$ lose one variable (the result of replacing two variables with a single one), and reduce the provenance by two monomials ($147.6\cdot SB\cdot m_1+ 181.15\cdot SB\cdot m_3$ instead of $77.9\cdot b_1\cdot m_1+ 80.5\cdot b_1\cdot m_3+69.7\cdot b_2\cdot m_1+ 100.65\cdot b_2\cdot m_3$ in $P_2$). Similarly, the in the array of $Sp$ (we use $Std$, $B$ and $Sp$ as shorthand for $Standard$, $Business$ and $Special$ respectively), $A_{Sp}[4]= 2$. For node $B$, we obtain the array $A_B = [0, \bot, 1, \bot, 2, \bot]$. $A_B[2] = 1$ is obtained by using $B$'s children arrays, since $A_{SB}[2] = 1$ and $A_{e}[0] = 0$, and $A_B[4] = 2$ by the abstraction that uses $B$. 
		Finally, $A_{Plans}$ is computed using $A_B$, $A_{Sp}$ and $A_{St}$, and the resulting array is $A_{Plans} = [0, \bot, 1, \bot, 2, 3]$. Here, the value $A_{Plans}[5] = 3$ is obtained using $A_B[2] = 1$ and $A_B[4] = 2$. The corresponding valid variable set is $\{SB, Sp,e,p_1\}$,  the monomial loss using this abstraction is 6 and the variable loss is~3.
	
\end{example}

\begin{proposition}[Complexity]
\label{prop:comp}
Given a set of polynomials $\mathcal{P}$, variable abstraction tree $T$, and a bound $B$, the complexity of Algorithm \ref{algo:varSelection} is $O(n\cdot w \cdot k^2\cdot \numM{\mathcal{P}})$ where $n$ is the number of nodes in $T$, $w$ is the width of $T$, and $k = \numM{\mathcal{P}} - B$.
\end{proposition}

Note that $k$ is bounded by the size of the given provenance polynomial, so the algorithm is in PTIME.

\subsection{Multiple Abstraction Trees} \label{sec:greedy}

Compressing a (multi)set of polynomials $\mathcal{P}$ polynomial using  
an arbitrary abstraction  forest $\mathcal{T}$ 
 in an optimal manner is intractable (see \Cref{prop:ptime}).
Thus, we provide in Algorithm \ref{algo:greedy}  a greedy heuristic solution. 

Roughly speaking, Algorithm \ref{algo:greedy} works as follows:
It maintains a VVS $S$ and keeps updating it and shrinking the provenance size until $S$ becomes adequate for $B$.
The set is initialized to contain all the leaves in the forest.
Then, the algorithm iteratively replaces a set of sibling nodes in $S$ with their parents. To make this process more efficient, 
the algorithm 
maintains a set $C$ of \emph{candidate} parents which are not in $S$, but all their children are. The greedy aspect of the algorithm is that the candidate parent it choses is the one which entails the minimal variable loss.
%

In more detail, given 
bound $B$,  Algorithm \ref{algo:greedy} first initializes the current VVS $S$ and candidate nodes to $\emptyset$ (line \ref{line:greedyInit1}) and $k$ to $\numM{\mathcal{P}} - B$ (line \ref{line:greedyInit3}). Then, in lines \ref{line:greedyAddleaves1}--\ref{line:greedyAddleaves2} the algorithm adds the leaves of all the trees in $\mathcal{T}$ to $S$ and updates the candidates list accordingly in lines \ref{line:greedyAddToC1}--\ref{line:greedyAddToC2}.

\SetInd{1.3ex}{1.3ex}
\begin{algorithm}
	\SetKwInOut{Input}{input}\SetKwInOut{Output}{output}
	\SetKwFunction{ML}{ML}\SetKwFunction{VL}{VL}
	\SetKwFunction{allRoots}{allRoots}
	\LinesNumbered
	\Input{A set of polynomials $\mathcal{P}$, a set of variable abstraction forest $\mathcal{T}$, and bound $B$}
	\Output{Valid variables set $S$} \BlankLine
	{
		$S \gets \emptyset$; $C \gets \emptyset$\;\label{line:greedyInit1}
		$k \gets \numM{\mathcal{P}} - B$\; \label{line:greedyInit3}
		
		\ForEach{tree $T$ in $\mathcal{T}$}
		{\label{line:greedyAddleaves1}
			\ForEach{leaf $l$ in $T$}
			{
				$S \gets S\cup \{l\}$ \label{line:greedyAddleaves2}
			}
		}
		\ForEach{tree $T$ in $\mathcal{T}$}
		{\label{line:greedyAddToC1}
			\ForEach{node $n$ in $T$}
			{
				\If {$n.children \subseteq S$}
				{
					$C \gets C\cup \{n\}$ \label{line:greedyAddToC2}
				}
			}
		}
		\While{ $\ML(S) < k$ and $C \neq \emptyset$}
		{\label{line:greedyWhileLoopStart}
			\nonl let $c \in C$ be the node with minimal $\VL((S\setminus c.chilren)\cup  \{c\})$\;
			$C \gets C\setminus \{c\}$;$S \gets S\cup \{c\}$\; \label{line:greedyUpdateSets}
			$S \gets S\setminus \{c.children\}$\; \label{line:greedyUpdateS}
			\If {$(c.parent).children \subseteq S$}
			{ \label{line:greedyUpdateC}
				$C \gets C\cup \{c.parent\}$\; \label{line:greedyUpdateC2}
			}
			 \label{line:greedyWhileLoopEnd}
			
		}

		\Return $S$\; \label{line:greedyReturn}
	}
	\caption{Greedy Valid Variables Selection}\label{algo:greedy}
\end{algorithm} 

Then, while the monomial loss of $S$ is less than $k$ and the candidate set is not empty (lines \ref{line:greedyWhileLoopStart} -- \ref{line:greedyWhileLoopEnd}), the algorithm selects the candidate $c$ such that adding $c$ to $S$ results in minimal variables loss with respect to the polynomials obtained by the abstraction $S$ (ties are broken arbitrarily), adds it to $S$ and removes it from $C$ (line \ref{line:greedyUpdateSets}). We then remove $c$'s children from $S$ (line \ref{line:greedyUpdateS}) and if the parent of $c$ is now a candidate (line \ref{line:greedyUpdateC}), add it to $C$. 

\mypara{Complexity} Given a set of polynomials $\mathcal{P}$, a variable abstraction forest $\mathcal{T}$, and a bound $B$, the complexity of Algorithm
\ref{algo:greedy} is $O(n\cdot \numM{\mathcal{P}})$ where $n$ is the  number of nodes in $\mathcal{T}$. 
We experimentally study the preference of the algorithm and the quality of its results. Naturally, the results accuracy depend on the abstraction tree structure, however, we show that in many cases we obtain high quality results (see \S\ref{sec:exp}).

\begin{example}
Consider again the polynomials $\mathcal{P} = \{P_1, P_2\}$ from Example \ref{ex:optAlgoRun}. Given $\mathcal{P}$, the abstraction trees (after removing redundant nodes) $\mathcal{T} = \{Plans, Year\}$  shown in Figures \ref{fig:plansLattice} and \ref{fig:monthLattice}, and a bound $B = 4$, the greedy algorithm first compute $k \gets \numM{\mathcal{P}} - B = 14 - 4 = 10$, initializes $S$ with the leaves of all trees, and the set of candidates to be $C = \{SB,Sp, q_1\}$. $VL(Sp) = 2$, and both $SB$ and $q_1$ have the same $VL$ of 1, but the $ML$ of the VVS obtained from $S$ by adding $q_1$ is 7, whereas the variable loss obtained from $SB$ is 2. The algorithm selects $q_1$, which is then removed from $C$ and added to $S$; $m_1$ and $m_3$ are removed from $S$. Since the monomial loss of $S$ is 7 (less than $k = 10$) it continues to the next iteration, where the candidate $SB$ is selected. When $SB$ is added to $S$, the node $B$ is added to the candidates list $C$ since all of its children are in $S$. The $VL$ of $B$ is 1, since the polynomials obtained using $S$ contains the variables $SB$ and $e$. The monomial loss of the VVS obtained by adding $SB$ is 8. Next, $B$ is added to $S$, $e$ and $SB$ are removed, and the $ML$ of $S$ is 9. Finally, the node $Sp$ is added to $S$, and $f_1$, $y_1$ and $v$ are removed resulting in a VVS $S$ with $ML = 11 > k$ and the algorithm terminates with $VL$ of 5. Note that in this example the VVS $\{q_1, Sp, SB,e,p_1\}$ is the optimal VVS with $ML = 10$ and $VL=4$.
 \end{example}

%% file: exp.tex
\section{Experimental Results}
\label{sec:exp}

We  implemented the algorithms and 
experimentally evaluated their scalability
and usefulness 
in various settings. 



\subsection{Implementation and Optimizations}

The system is implemented in Python 3. We used Any Python Tree Data package \cite{anytree} to represent the abstraction trees, and Python's dictionaries for the polynomials. To allow for scalability of the algorithms, in terms of
their execution time,
we have employed optimizations in different
parts of the implementation, and we explain them next.

\mypara{Efficient ML computation}
Given the polynomials $\mathcal{P}$ and an abstraction tree $T$, a naive way to compute the monomial loss of a node $v\in T$ is as follows: let $v_0,\ldots, v_m$ be the descendants of $v$ in $T$. For every polynomial $P \in \mathcal{P}$, replace every occurrence of any $v_i$ for $0\leq i \leq m$ by $v$ to obtain $P_v$, and then compute $ML(v) = \sum_{P \in \mathcal{P}}\numM{P}-\numM{P_v}$. We note that this computation requires traversing the polynomials for every node in the abstraction trees. For a single tree, computing the monomial loss for every tree node, may be done by traversing the polynomials only once as a first step of the algorithm as follows. Let $l$ be a leaf of $T$, and $M$ be a monomial in the polynomial $P$ where $M$ contains the variable $l$. We denote by $M_l$ the monomial obtained by removing $l$ from $M$. For every polynomial $P$, we use a hash table $D_{P}$ that maps each leaf $l \in T$ to a set $\{M_l\mid M \text{ is a monomial in } P \text{ and } l \text{ appears in } M\}$. Then, the monomials loss of a node $v$ in $T$, where $l_0,\ldots,l_m$ are the leaves descendants of $v$, is 
$$\sum_{i=0}^m|D_{P}[l_i]|-|\bigcup_{i=0}^m D_{P}[l_i]|$$   
Intuitively, the monomials in $\{M \mid M_{_i}\in D_{P}[l_i], 0 \leq i \leq m\}$ are the monomials affected by replacing $l_0,\ldots,l_m$ with $v$, and their number is $\sum_{i=0}^m|D_{P}[l_i]|$. Every element $e$ that appears in more than one $D_{P}[l_i]$ for $0 \leq i \leq m$ is replaced by a single monomial $e\cdot v$ (with a different coefficient). Thus the number of different monomials obtained by the replacement is the number of distinct elements in $\bigcup_{i=0}^m D_{P}[l_i]$. The monomial loss is 
the difference between the total number of affected monomials and the number of remaining monomials.


\mypara{Optimizing $A_v$ computation}  Recall the procedure\\ \texttt{computeArray} in Algorithm \ref{algo:varSelection}, that computes the array $A_v$ for each node $v$. The procedure computes the array $A_v$ of a node $v\in T$ using the arrays $A_{v_1},\ldots, A_{v_m}$, where $v_1,\ldots, v_m$ are $v$'s children in $T$. The computation is done in a dynamic programming fashion, and includes a loop over the entries of each one of the arrays $v_1,\ldots, v_m$, each of size $\numM{\mathcal{P}}-B+1$. 

In practice, most of the entries of the arrays contain $\bot$; thus, instead of using arrays to represent $A_v$ we used hash tables, such that if the $i$ entry of the array is $j\neq\bot$, then the table contains the key $i$ with the value $j$. Now, the procedure loops over the hash keys (i.e., only entries that are not $\bot$), which is usually significantly smaller than  $\numM{\mathcal{P}}-B+1$. 
Additionally, note that a valid variable set in a subtree of a node $v\in T$ of height 1 may contain either $v$'s children or $v$. In the former case both the $VL$ and $ML$ are 0; thus we add the key $ML(v)$ with the value $VL(v)$ to $A_v$ for every node $v$ of height 1 in $T$ 
instead of 
invoking 
\texttt{computeArray} on it. 

\subsection{Benchmarking}
\label{sec:benchmarks}
We have developed a dedicated benchmark that involves
both synthetic and real data as follows.



   
\mypara{Telephony Company benchmark} We used the provenance generated for the query from our running example, where the plans price was parametrized by month and plan (by 12 and 128 variables respectively). The tables were populated with randomly generated date as follows. For each customer select randomly one of 128 possible plans, 5 digit zip code and the total number of calls durations for each month. We varied the number of customers from 10K to 5M, the latter leading to an input database of 65M tuples. The resulting provenance contained 100,000 polynomials, each one contains combinations of the 128 plans variables and 12 month variables, with total size of over  8GB.


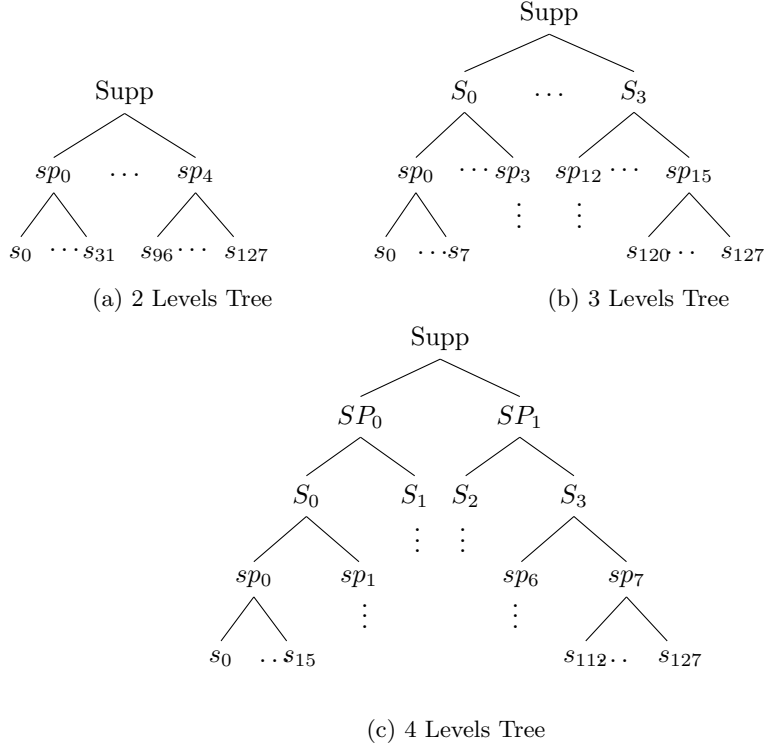
\begin{figure}
	\centering
	\begin{subfigure}[b]{0.4\linewidth}
		\begin{tikzpicture}
		\Tree [.Supp 
		[.$sp_0$ $s_{0}$ \quad $s_{31}$ ] 
		[.$sp_4$ $s_{96}$ \quad $s_{127}$ ]
		]
		\node (x) at (0,-1) {$\ldots$} ;	
		\node (x) at (-0.8,-2) {$\ldots$} ;	
		\node (x) at (0.9,-2) {$\ldots$} ;	
		\end{tikzpicture}
		\caption{2 Levels Tree}
		\label{fig:t1}
	\end{subfigure}%
	\begin{subfigure}[b]{0.6\linewidth}
		\begin{tikzpicture}
		\Tree [.Supp 
		[.$S_0$ 
		[.$sp_0$ $s_{0}$ \quad $s_{7}$ ] $sp_{3}$
		]
		[.$S_3$ $sp_{12}$
		[.$sp_{15}$ $s_{120}$ \quad $s_{127}$ ] ]
		]
		\node (x) at (0,-1) {$\ldots$} ;	
		\node (x) at (-1,-2) {$\ldots$} ;	
		\node (x) at (1,-2) {$\ldots$} ;	
		\node (x) at (-1.55,-3.1) {$\ldots$} ;	
		\node (x) at (1.75,-3.1) {$\ldots$} ;		
		\node (x) at (0.4,-2.5) {$\vdots$} ;		
		\node (x) at (-0.4,-2.5) {$\vdots$} ;		
		\end{tikzpicture}
		\caption{3 Levels Tree}
		\label{fig:t2}
	\end{subfigure}%
	
	\begin{subfigure}[b]{\linewidth}
		\begin{center}
			\begin{tikzpicture}
			\Tree [.Supp 
					[.$SP_0$ 
						[.$S_0$ 
							[.$sp_0$ $s_{0}$ \quad $s_{15}$ ] 
						$sp_{1}$
						] 
						[.$S_1$ ]
					]
					[.$SP_1$  
						$S_2$ 
						[.$S_3$ 
						$sp_{6}$
							[.$sp_{7}$ $s_{112}$ \quad $s_{127}$ ] 
						] 
					]
				   ]
			\node (x) at (-1,-3.5) {$\vdots$} ;		
			\node (x) at (1,-3.5) {$\vdots$} ;		
			\node (x) at (-0.3,-2.5) {$\vdots$} ;		
			\node (x) at (0.3, -2.5) {$\vdots$} ;		

			\node (x) at (-2.2,-4.2) {$\ldots$} ;	
			\node (x) at (2.3,-4.2) {$\ldots$} ;	
			\end{tikzpicture}
		\end{center}
		\caption{4 Levels Tree}
		\label{fig:t5}
	\end{subfigure}%
	\caption{TPC-H supplier abstraction trees}
	\label{fig:tpchTrees}
\end{figure}

\mypara{TPC-H benchmark}  The TPC Benchmark H (TPC-H)\cite{tpch} consists of a suite of business oriented queries. We have experimented with all non-nested TPC-H queries, a total of 8 queries (note that handling nested queries requires an extended provenance model, beyond polynomials \cite{AggPaper}), with all standard TPC-H settings. Due to space limitations, we report the results for three of the queries, namely Q1, Q5 and Q10 which are representative in the sense that they yield a large number of provenance polynomials, each containing a large number of monomials; the observed trends for the other queries were similar.   

Using 10GB input data, the provenance generated by query 5 consists of 25 polynomials, where the maximal polynomial contains 10890 monomials, and the minimal polynomial contains 10772 monomials (average size of 10839.8 monomials). In the provenance generated by query 10 the number of polynomials was 993306, where the largest polynomials contains 67 monomials, and the smallest only 2 monomials (average of 15.78 monomials per polynomial). Finally, for query 1 we obtain 8 polynomials, each one of size 11265.

For the choice of parameterization (i.e. where to place variables) there is no existing benchmark to our knowledge. We introduced suppliers variables $s_i$ and parts $p_i$ variables for $0 \leq i \leq 127$, and parameterized the discount attribute of the \texttt{LINEITEMS} table based on the \texttt{SUPPKEY} and \texttt{PARTKEY} attributes, where we used the variable $s_i$ if the suppliers key $k \mod 128 = i$, and similarly for the parts variable $p_j$. In realistic scenarios, assuming the user defines the variables and trees manually, the expected abstraction trees are likely to have no more than 128 leaves. 



\mypara[]{Abstraction trees} We generated 7 different abstraction tree structures for each set of variables (plans, suppliers and parts), with growing number of valid variables sets. Abstraction tree of type 1 is a 2 level tree as shown in Figure \ref{fig:t1}. We varied the number of the inner nodes $sp_i$ from 2 to 64. type 2, 3, and 4 trees are 3-level trees (see Figure \ref{fig:t2}), with root fan-out of 2, 4, and 8 respectively. Finally, the trees of type 5, 6, and 7 are 3-level trees as depicted in Figure \ref{fig:t5}, where in type 5 trees the root have 2 children ($SP_i$) each one with 2 children ($S_i$), type 6 trees have 2 $SP_i$ nodes each one with 4 children, and in type 7 trees the root has 4 children, and each one of them has 2 children. The number of $sp_i$ nodes varied from 2 and up to 16 (see Table \ref{tbl:treesTypes} in the Appendix).

We have conducted experiments for
examining the scalability of the approach, its usefulness in terms of the speedup in the time it
takes to use the compressed provenance polynomials compared with the uncompressed polynomials, and the quality of the greedy algorithm, in terms of accuracy and compression time.
The experiments were executed on Windows 10, 64-bit, with
8GB of RAM and Intel Core i7-4600U 2.10 GHz processor.

\subsection{Experimental Results}\label{subsec:results}

\begin{figure*}
	\centering
	\includegraphics[width = \linewidth]{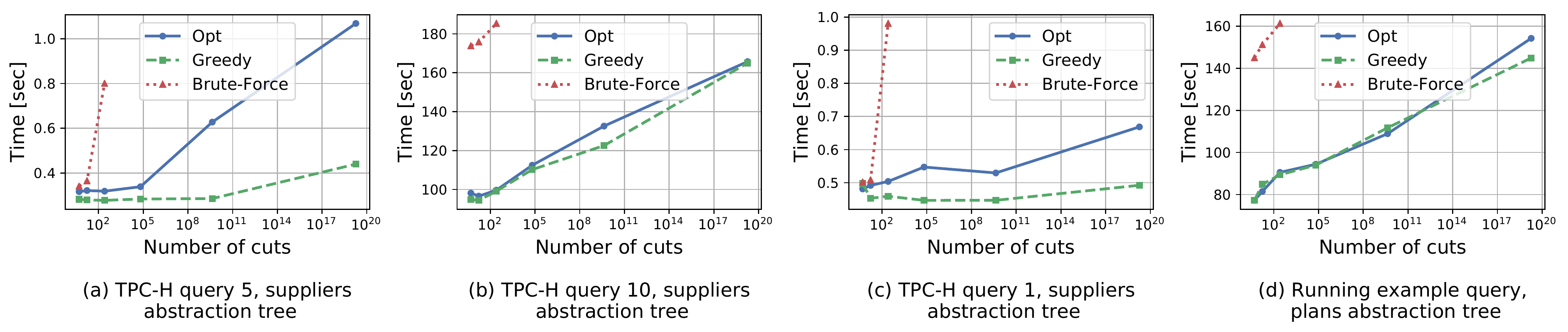}
	\caption{Compression Time as a Function of Valid Variables Sets for 2 Levels Tree (tree type 1)}
	\label{fig:compType1}
\end{figure*}


\begin{figure*}
	\centering
	\includegraphics[width = \linewidth]{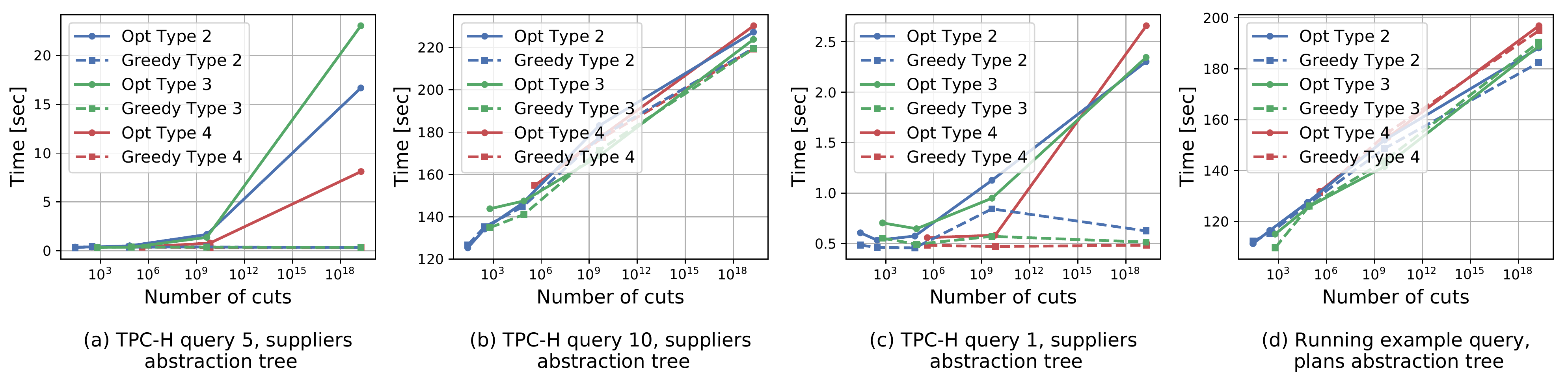}
	\caption{Compression Time as a Function of Valid Variables Sets for 3 Levels Tree}
	\label{fig:compType2}
\end{figure*}


\begin{figure*}
	\centering
	\includegraphics[width = \linewidth]{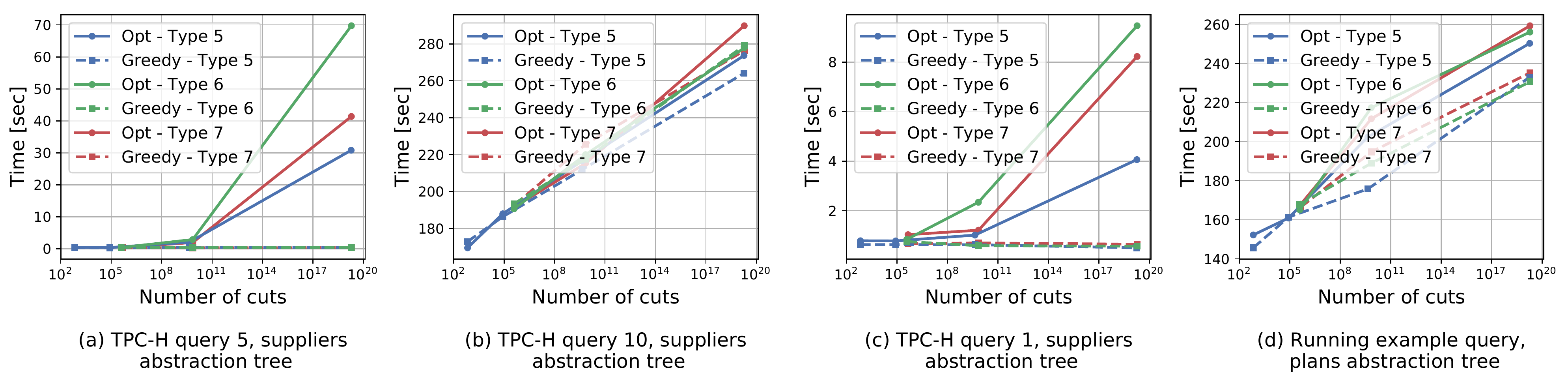}
	\caption{Compression Time as a Function of Valid Variables Sets for 4 Levels Tree}
	\label{fig:compType3}
\end{figure*}

\begin{figure*}
	\centering
		\includegraphics[width = \linewidth]{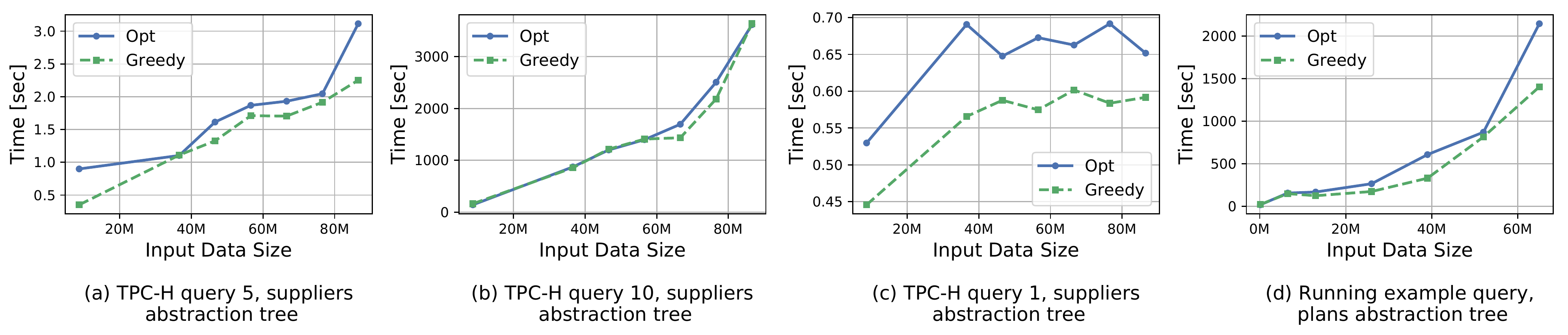}
		\caption{Compression Time as a Function of the Input Data Size}
			\label{fig:compInputDB}
\end{figure*}


\begin{figure*}
	\centering
	\includegraphics[width = \linewidth]{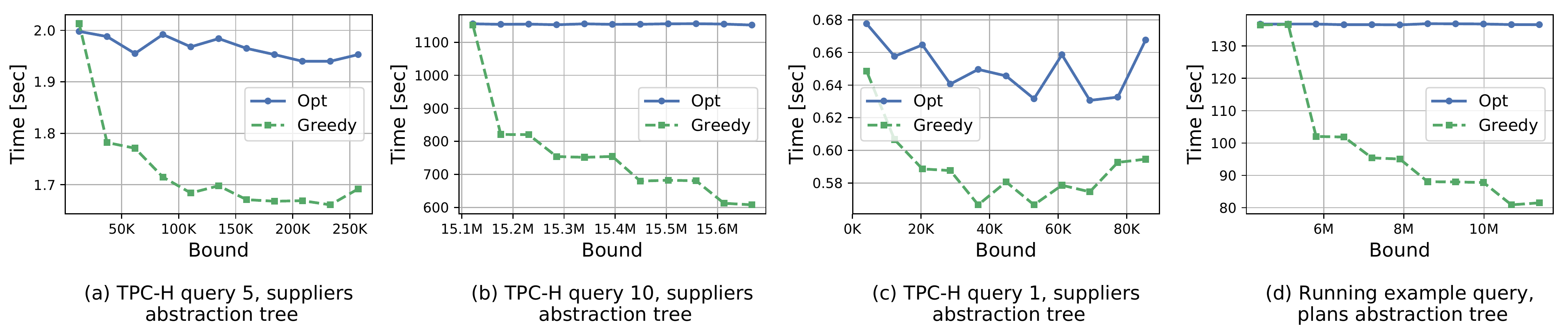}
	\caption{Compression Time as a Function of Bound}
	\label{fig:bound}
\end{figure*}


\begin{figure*}
	\centering
    \includegraphics[width = \linewidth]{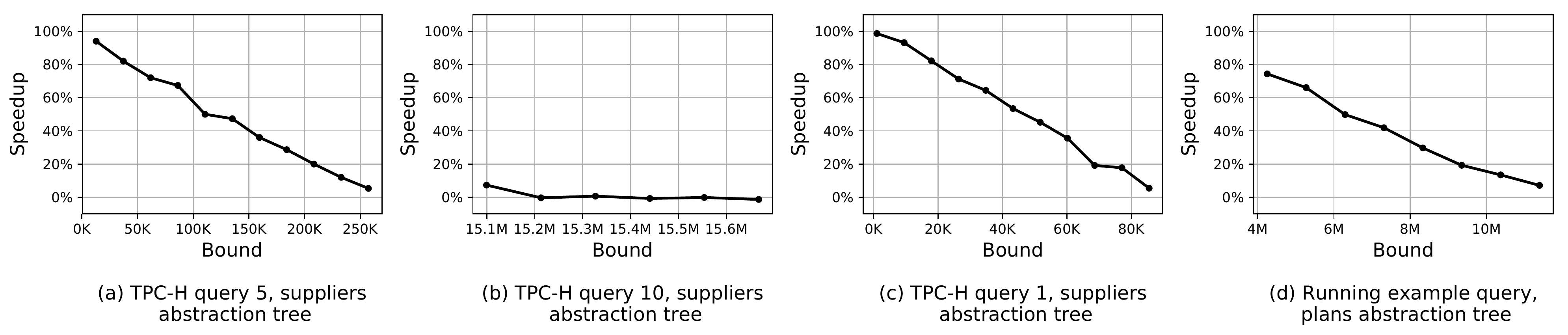}
	\caption{Assignment Time Speedup as a Function of Bound}
	\label{fig:assTime}
\end{figure*}


\begin{figure*}
	\centering
		\includegraphics[width = \linewidth]{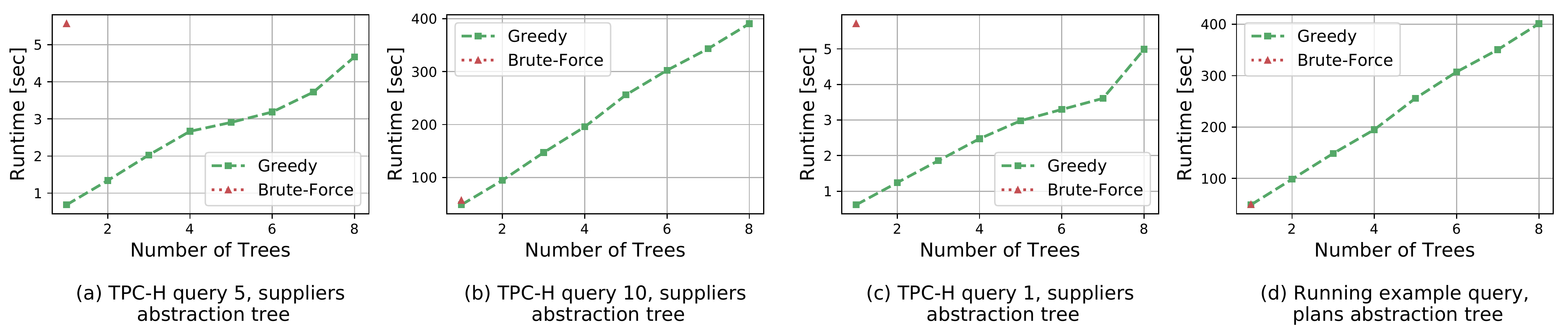}
	\caption{Compression Time as a Function of Number of Trees}
	\label{fig:compNumOfTrees}
\end{figure*}

\begin{table}\footnotesize

	\begin{tabular}{c}
	\begin{tabular}{cc}\hline
		TCPH query 5& TCPH query 10\\
		\begin{tabular}{ccc}
		Tree Type & Accuracy & Speedup \\ \hline
		 1 &100\% & 20.44\% \\
		 2& 95.26\% &37.72\%\\
		 3& 90.46\%&45.8\%\\
		 4& 88.32\%&50.22\%\\
		 5& 90.99\%&49.51\%\\
		 6& 87.16\%&63.52\%\\
		 7& 87.06\%&61.71\%\\ 
		\end{tabular}
		 &
		 \begin{tabular}{ccc}
		Tree Type & Accuracy & Speedup \\ \hline
		   1 &100\% & 28.03\% \\
		 2& 81.39\% &30.46\%\\
		 3& 74.78\%&36.96\%\\
		 4& 64.7\%&36.61\%\\
		 5& 68.74\%&37.64\%\\
		 6& 65.01\%&34.67\%\\
		 7& 55.95\%&41.72\%\\
		\end{tabular}\\
		\hline
	\end{tabular}
	\\
	\\
	\begin{tabular}{cc}\hline
	TCPH query 1& Running example\\
	 \begin{tabular}{ccc}
			Tree Type & Accuracy & Speedup \\ \hline
		 1 &100\% & 9.18\% \\
		 2& 95.94\% &26.02\%\\
		 3& 92.32\%&30.72\%\\
		 4& 83.38\%&36.67\%\\
		 5& 90.99\%&39.3\%\\
		 6& 89.33\%&43\%\\
		 7& 88.41\%&45.94\%\\
		\end{tabular}
		&
	
		\begin{tabular}{ccc}
		Tree Type & Accuracy & Speedup \\ \hline
		 1 &99.84\% & 21.22\% \\
		 2& 78.58\% &21.23\%\\
		 3& 62.74\%&27.86\%\\
		 4& 65.4\%&30.95\%\\
		 5& 61.95\%&27.49\%\\
		 6& 65.79\%&23.78\%\\
		 7& 57.77\%&30.89\%\\  
	\end{tabular}\\\hline
	\end{tabular}
	\end{tabular}
	\caption{Greedy algorithm average accuracy and speedup}
	\label{tbl:greedy_acc}
\end{table}

The first set of experiments aims at studying the provenance compression time. In the following, we refer to Algorithm \ref{algo:varSelection} (performing exact computation) as Opt VVS, and Algorithm \ref{algo:greedy} as the greedy algorithm. We have also compared the performance of our algorithms to a baseline of a brute force algorithm, that loops over all possible VVS and selects the optimal one. We examine the execution time of the algorithm for each dataset, and for each one of the abstraction tree, using a single tree in each execution. For all experiments we have set the bound to be 0.5 of the input polynomials size, except for those studying the effect of the bound itself. We present the results for the execution using the supplier abstraction tree for the TPC-H queries, the results obtained using the parts abstraction tree showed similar trends.

\mypara{Tree Structure}
Figures \ref{fig:compType1} -- \ref{fig:compType3} presents the compression time as a function of the number of valid variables set for different tree's types. We observed a moderate growth in the compression time using Opt VVS (solid lines) and the greedy algorithm (dashed lines) in all cases, while the brute force algorithm (dotted lines) was able to complete the computation only when the number of VVS was less than 80,000. For TPC-H queries 5 and 1 (Figures \ref{fig:compType1}a and \ref{fig:compType1}c), the greedy algorithm was up to 22\% faster than the optimal algorithm and for TPC-H query 10 and the running example query (Figures \ref{fig:compType1}b and \ref{fig:compType1}d), 
their running time  was similar (note, though,  the different scales).
This is because 
the required compression in running example can be obtain only using the root, and for TPC-H query 10 there is no abstraction that lead to the desired compression, thus in the the latter queries   the greedy algorithm have to traverse the whole tree, whereas for queries 1 and 5 it halts when the required provenance size is reached (without going over all the tree nodes).
Similar results were obtained for 3 level trees (Figure \ref{fig:compType2}) and 4 level trees (Figure \ref{fig:compType3}). The results of the brute force algorithm (omitted from the graphs) were above 559 seconds for type 2 tree with 66050 VVS for query, 595 seconds for query 5, 756 for query 10 and 737 for the running example query. 

\mypara{Data Size}
Figure \ref{fig:compInputDB}   presents the provenance compression time as a function of the input data size (number of tuples). We used TPC-H 1Gb and growing fragment size of the 10Gb data set. Using about 40M tuples, every polynomial in the provenance polynomials generated by query 1 contains all possible combination of parts and suppliers variables. This is due to the fact that the number of polynomials generated by the query is relatively small (only 8).  Thus the polynomials generated using larger fragments of the data are similar, up to the coefficients, and the computation time is similar from that point onwards as shown in Figure \ref{fig:compInputDB}c. The greedy algorithm outperforms Opt VVS in this case by 15\% in average. For the other queries, we observed a moderate growth in the computation time with respect to the input data size.
 

\mypara[]{Bound}
Figure \ref{fig:bound} depicts the compression time as a function of the bound. We used the abstraction trees to compute the maximal and minimal possible compression bound for the provenance generated by each one of the queries. We then varied the bound to examine the effect of the bound on the computation time of the algorithms. The execution time of the Opt VVS is not affected by the bound, while the greedy algorithm execution time decreases as the bound increases, showing that it effectively exploits cases where a solution may be found by traversing less nodes.

\mypara[]{Compression Gain} The second set of experiments aims at assessing the usefulness of the approach: it studies the time it takes to use the compressed provenance for observing results under hypothetical scenarios, compared with the time of the original provenance expression. Figure \ref{fig:assTime} shows the assignment time speedup as a function of the compression bound. The for queries 1, 5 was up to 100\% and just below 80\% for the running example. The provenance resulting from query 10 consist of a large number of polynomials (about 100,000) with small number of monomials (15 in average). 
Here, the maximal possible compression is relatively small (about 0.03\%) and the speedup is negligible.  

\mypara[]{Greedy Algorithm} 
This set of experiments aims at assessing the performance of the greedy algorithm and the quality of its results. We used a set of eight (3-level) binary trees, each with 16 leaf. Each tree contains 16 out of 128 variables. We then examine the affect of the number of possible trees on the execution time of the greedy algorithm. The results are shown in Figure \ref{fig:compNumOfTrees}. We observed a moderate growth for all the queries.
%
 %
We then examine the quality of the greedy algorithm results. We compared the $VL$ of the resulting VVS generated by the greedy algorithm to the $VL$ of the VVS computed by Opt VVS using. Table \ref{tbl:greedy_acc} depicts the accuracy and speedup compared to Opt VVS for the different queries and trees. The accuracy for Type 1 trees was 100\% in most cases; this is because type 1 trees are 1 level trees, thus if the optimal VVS is not the root, it contains leaves and middle level nodes, the former are added to the initialized greedy VVS and the latter are inserted to the candidates list. Since the number of leaves of each middle level node are equal in those trees, the $VL$ of all of them is equal, and in most cases, arbitrary selection and insertion of the middle level nodes leads to an optimal solution. 
In general, the results of the greedy algorithm were more precise for TCP-H queries 1 and 5. This is due to the fact that the provenance resulting from those queries contains significantly smaller number of polynomials, 8 and 25 respectively, compared to about 100,000 in both the provenance of TCP-H query 10 and the running example. Therefore, the latter are more sensitive to ``locally" greedy selection. Furthermore, the speedup of the greedy algorithm was between 9.18\% and up to 63.52\%.

\begin{figure}
	\centering
	\includegraphics[width = \linewidth]{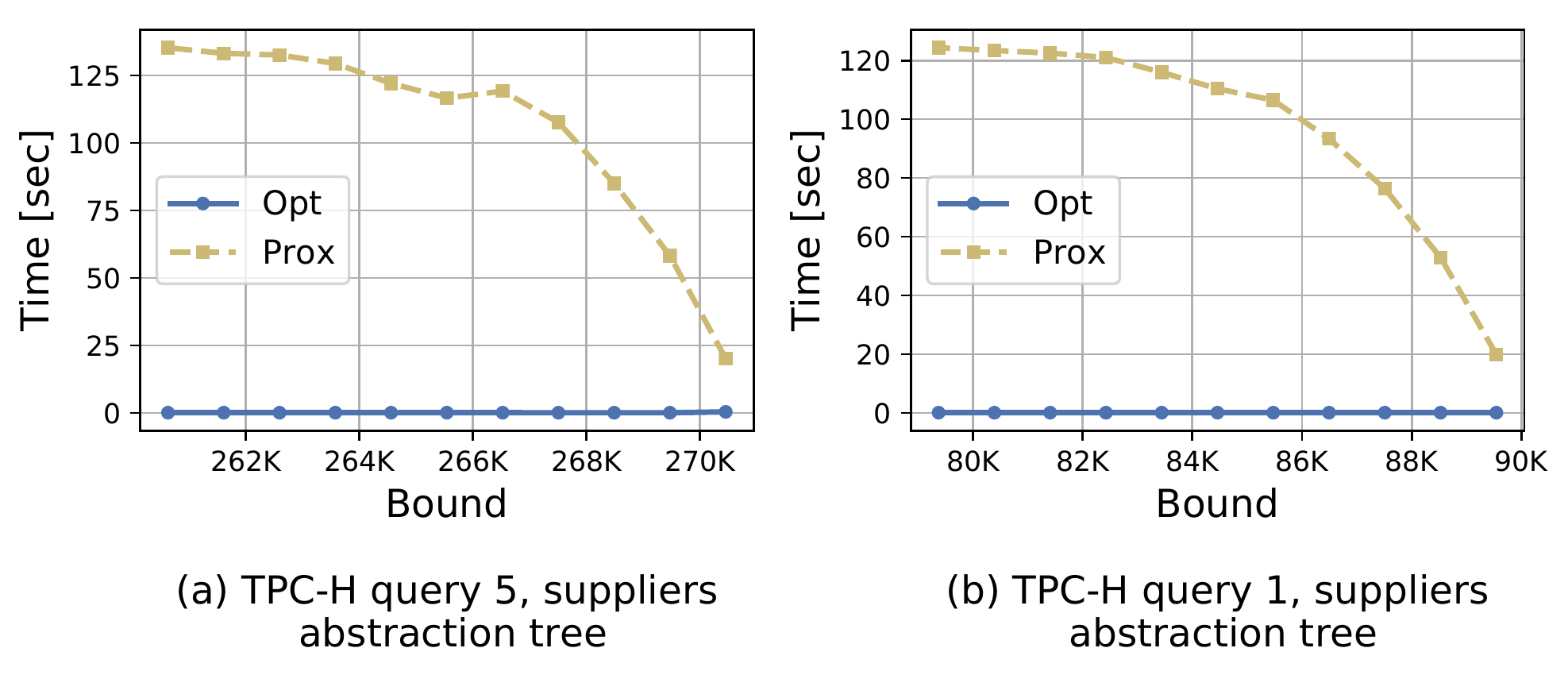}
	\caption{Compression Time as a Function of Bound}
	\label{fig:prox_exp}
\end{figure}

\mypara[]{Gain of abstraction trees} 
As mentioned in \S\ref{sec:intro}, there is previous work on summarization guided by various optimization tasks. Most generally, the algorithm presented in \cite{prox} 
uses an oracle to define semantic constraints over the summarizations and their gain. We have implemented the algorithm proposed by \cite{prox} and used for each experimental case the relevant abstraction tree as the black-box oracle that determines the semantic constraints, provenance size and candidate scores (using the terminology of \cite{prox}). Note that the algorithm of \cite{prox} gets as input a single polynomial; 
thus,
we extend it to work on a set of polynomials. The run-times of this competitor algorithm compared to our optimal algorithm as a function of the bound (using the setting described in the bound experiment) for TPC-H queries 1 and 5 are presented in Figure \ref{fig:prox_exp}. The competitor algorithm did not finish the computation on query 10 and the running example query within 24 hours; we thus observe the crucial advantage of our optimal algorithm in terms of runtime. 

The algorithm of \cite{prox} iteratively examine, using the oracle, the grouping of all possible monomial's pair in the provenance polynomials in order to reduce its size with minimal loss. Thus, as the bound decreases, the number of oracle calls, and the algorithm's total runtime increase. 
In terms of accuracy, unlike our solution the algorithm of \cite{prox} has no quality guarantees. In practice, for the cases where convergence was reached, it has achieved good accuracy, though still inferior to the optimal solution that our algorithm finds (the quality of \cite{prox} was approx. 96\% of the optimal on average).

%% file: related.tex
\section{Related Work}
\label{sec:related}


\mypara{Data Provenance} 
Data provenance has been studied for different query languages, from relational algebra to Nested Relational Calculus, and
with different provenance models (see, e.g.,
\cite{trio,GKT-pods07,GlavicSAM13,GeertsP10,KenigGS13,CheneyProvenance,w3c,ProvenanceBuneman,Olteanu12})
and applications \cite{Suciu,Meliou2,RoyS14,Gatt,GlavicAMH10,
	icde16}. We focus on this work on provenance polynomials, an approach that originated in the provenance semiring approach \cite{GKT-pods07} for SPJU queries and have since been used in for further expressive queries, and for applications ranging from hypothetical reasoning and reenactment to explanations of missing answers (e.g.,  \cite{GKT-pods07, cidr13, AggPaper, FosterGT08, XuZAT17,BidoitHT15,glaviccikm}). Furthermore, the work of \cite{ICDT09-containment} has shown that the provenance semiring framework is expressive enough so that provenance polynomials can capture many of the previously proposed provenance models -- this intuitively means that we assign different semantics to the addition and multiplication operations. The polynomial model that we work with in this paper is generic, in the sense that we assume no concrete semantics of the operations. Our solution may thus be used for abstraction of provenance in a wide range of applications and using a wide range of provenance models.


\mypara{Hypothetical Reasoning} 
Answering queries under hypothetical updates was studied in \cite{GhandeharizadehHJ96,BalminPP00,cidr13,valICDT16, vldbj,glaviccikm,ArabGKRG18}.
This work has laid the grounds for an efficient way of
reasoning with hypothetical updates, namely that of 
\emph{provisioning}. The solution employed in \cite{valICDT16} uses provenance polynomials of the flavor we have studied in this work, but while raising the need for specification of scenarios that can guide provenance generation, no such models were proposed there. Different notions of hypothetical reasoning vary in the level of detail that they track (from fine-grained provenance tracking required for hypothetical reasoning to coarse grained tracking of, e.g., the modules that were used), and the scale of data they can handle; there is an obvious tradeoff between these factors. In the present paper we formalize and study, for the first time to our knowledge, this tradoff, and present the notion of abstraction trees to control it.

\mypara{Provenance Storage}
Multiple lines of work have focused on efficient provenance storage (e.g.,
\cite{Bao,Chapman,OlteanuZ12,Olteanu12, icdt14}) including support for distributed
systems \cite{ramp,lipstick,IkedaPW11,Titian,ChenLLAJLZ17}. In contrast to our work, none of these solutions has focused on use-based approximation of the provenance, namely to leverage the expected uses -- restrictions on hypothetical reasoning, in our case -- to abstract the provenance expression. The efficient storage techniques described above are thus complementary to our work; implementing our approach for summarization to work in tandem with solutions for efficient storage is an important goal for future work.  Specifically, in either centralized or distributed settings, once provenance is computed we may wish to interact with it -- i.e., assign values -- using a single, weak machine such as the analyst's phone/tablet, thus avoiding costly storage; this is where our abstraction would come into play.

\mypara{Provenance Size Reduction} Our solution aims at reducing provenance size through a novel notion of abstraction. There are multiple other lines of work aiming at reducing provenance size. One such approach is to track only parts of the provenance that are of interest, e.g., based on user specification \cite{glavic2013using,GlavicICDE09,selP,CheneyAA14}. Such selective provenance tracking is complementary to ours, as we start from a polynomial which may reflect full or partial provenance, invariably for our approach. Other works have studied techniques for the factorization or summarization of provenance, for compression. Such summarization may be lossless \cite{Olteanu12, icdt14}, and then it may again be used for storage in conjunction with our lossy approach. Lossy compression techniques have been proposed for multiple purposes, namely probability computation \cite{DBLP:journals/pvldb/ReS08} or the presentation of explanations \cite{LeeNLG17}. The only work, to our knowledge, that studied provenance compression for hypothetical reasoning is that of \cite{valICDT16} where algorithms are proposed to design sketches of the provenance that account for a given set of hypothetical scenarios. As mentioned in \cite{valICDT16}, their construction assumes that the hypothetical scenarios are given, but proposing means for specifying them has been left there for future work. 


%% file: conc.tex

\section{Conclusion and Future Work}
\label{sec:conc}

We have presented a novel solution 
for summarizing provenance by grouping together variables and replacing them with an abstract meta-variable. Grouping is constrained through abstraction trees, defining logical hierarchies over the variables. We have introduced the problem of optimizing abstractions that reduce the provenance size below a given threshold, and studied its computational complexity. We have proposed efficient algorithms and heuristics for the problem, and experimentally shown their effectiveness.    

There are many intriguing directions to explore as future extensions of the framework. 
One such promising direction is to combine our approach with other solutions for efficient provenance storage and generation surveyed in~\S\ref{sec:related}.
Arguably, the most intriguing direction for future work is enabling online compression:
Our solutions take  as input provenance polynomials; as explained in \S\ref{sec:intro}, the use case is that the provenance is generated once but then used multiple times, possibly by multiple analysts. 
Naturally, one may wonder whether we can 
avoid its 
costly computation to begin with, and instead compress provenance alongside with its generation?

One direction for using our approach for on-the-fly compression is through sampling. The idea is to generate only a sample of the provenance, apply our algorithms to the sample, and obtain a choice of Valid Variable Set (VVS). Then use the same VVS to group variables in the full input database, and generate provenance over this smaller set of ``summary" variables to obtain a summarized full provenance. The gaps are then (1) how to sample so that we correctly represent the data for this purpose, and (2) how to adapt the bound given to the algorithm when run over the sample. 

A naive approach for addressing the first challenge 
 is to uniformly sample from the data; depending on the query structure, this however may not lead to a representative sample of the output or its provenance. A heuristic solution, tailored for simple \texttt{GROUPBY} queries as in our running example, is to sample only from the relations that include the grouping attributes, leaving the other relations  intact. This intuitively results in a sample from the output polynomials, though this sample is still not guaranteed to be representative. Designing algorithms that achieve such a representative sample is an important challenge for future work.

As for adapting the bound to the sample, a possible heuristics is to set this bound as a function of (1) the original bound and (2) the ratio between the full provenance size and the sample provenance size, e.g. the first multiplied by the second. To this end, we need an estimation of the full provenance size. This is again a challenge, related to the classical problem of estimating queries output size (e.g.,~\cite{LiptonNSS93,ChenY17}). A possible heuristic is to perform multiple samples of increasing sizes, compute the provenance for each of them, and extrapolate the full provenance \cite{brezinski2013extrapolation}.  This component of provenance size estimation is 
an important challenge for future work.

%% file: hardness-appendix.tex

\newpage
 
\section{Lower Bound}\label{sec:hardness-app}

In this section we prove that 
the decision problem corresponding to provenance abstraction 
is  NP-hard in the provenance size (\Cref{prop:main}).
In fact, the proof shows that the problem is NP-hard even if we restrict our attention to abstraction trees of height two and to a single polynomial in which each monomial contains  exactly two variables.  Technically, the proof is done by a reduction from the vertex cover problem. 

\paragraph{Simplifying Assumptions.} 
Recall that the decision problem (\Cref{def:decisionProblem}) is as follows:
Given a set of polynomials $\mathcal{P}$, a valid abstraction forest $\mathcal{T}$ which is compatible with it, and desired size $B\in \{1 \ldots \numM{\mathcal{P}} \}$ and  granularity $K\in\{1 \ldots \numV{\mathcal{P}} \}$ of the set of polynomials
resulting from abstracting $\mathcal{P}$ according to $S$. 
Determine if there exists a precise valid variable set for $B$ and $K$.
For simplicity, and 
without loss of generality, we ignore in this section   
the trivial cases where either 
$B=1$, $B=\numM{\mathcal{P}}$, $K=1$, or $K=\numV{\mathcal{P}}$.


\paragraph{Notations.} 
For clarity, we wish to distinguish in this section between names of metavariables and names of variables.
Thus, we assume an infinite set $\id \in \ID$ of  \emph{identifier} which we use to form names: 
We generate names of metavariables by adding a parenthesized numerical superscript to a (possibly with a prime, a dot, a bar, etc.) identifier, e.g., $x^{(1)}$, $\dot{x}^{(1)}$, ${\bar x}^{(i)}$, and $\bar{\ddot{x}}^{(1)}$ are names of metavariables.
We generate variable names by adding one or two subscripts to names of metavariables, e.g., $x^{(1)}_2$, $\ddot{x}^{(1)}_k$, ${\bar x}^{(i)}_{a,2}$, and ${\bar{\dot{x}}}^{(i)}_{j,k}$ are names of variables.


\paragraph{Uniformly Partitioned Polynomials}
Technically, we prove that the decision problem is NP-Hard even if we  restrict our attention to a particular class of polynomials, dubbed \emph{uniformly partitioned polynomials}, which we now define.

\begin{definition}
A polynomial $P$ is \emph{uniformly partitioned according to a finite set $X$ of metavriables, a blowup factor $\bu \in \Nat$, and a set $I \subset \range{|X|}\times\range{|X|}$
},
denoted by $\upp{P}{X}{n}{I}$,
if for any $(a,b)\in I$ it holds that $a<b$ and  $P$ is of the form
\[
P \  = \sum_{(a,b)\in I} P^{(a,b)}  
\quad
\mathit{where}
\quad
P^{(a,b)} = \!\!\!\!\sum_{i,j \in \range{n}} x^{(a)}_i \cdot x^{(b)}_j \ .
\]
\end{definition}

\begin{example}\label{ex:uni-part-poly}
Assume $X=\{x^{(1)}, x^{(2)}, x^{(3)}, x^{(4)}\}$,
$n=3$, and $I=\{\}$.
The following polynomial $P$ is uniformly partitioned according to the set  $X=\{x^{(1)}, x^{(2)}, x^{(3)}, x^{(4)}\}$ of metavariables, 
the blowup factor $n=3$, and the set $I=\{(1,2),(1,3),(2,3),(2,4)\}$:
\[
 P = P^{(1,2)} + P^{(1,3)} + P^{(2,3)} + P^{(2,4)} ,  
\]
where
\[
\begin{array}{rcl}
P^{(1,2)} & = &
 x^{(1)}_1 \cdot x^{(2)}_1 + 	
 x^{(1)}_1 \cdot x^{(2)}_2 + 	
 x^{(1)}_1 \cdot x^{(2)}_3 + {}
 \\	
 &&
 x^{(1)}_2 \cdot x^{(2)}_1 + 	
 x^{(1)}_2 \cdot x^{(2)}_2 + 	
 x^{(1)}_2 \cdot x^{(2)}_3 + {}
 \\	
 &&
 x^{(1)}_3 \cdot x^{(2)}_1 + 	
 x^{(1)}_3 \cdot x^{(2)}_2 + 	
 x^{(1)}_3 \cdot x^{(2)}_3 
\end{array}
\]
\[
\begin{array}{rcl}
P^{(1,3)} & = &
 x^{(1)}_1 \cdot x^{(3)}_1 + 	
 x^{(1)}_1 \cdot x^{(3)}_2 + 	
 x^{(1)}_1 \cdot x^{(3)}_3 + {}
 \\	
 &&
 x^{(1)}_2 \cdot x^{(3)}_1 + 	
 x^{(1)}_2 \cdot x^{(3)}_2 + 	
 x^{(1)}_2 \cdot x^{(3)}_3 + {}
 \\	
 &&
 x^{(1)}_3 \cdot x^{(3)}_1 + 	
 x^{(1)}_3 \cdot x^{(3)}_2 + 	
 x^{(1)}_3 \cdot x^{(3)}_3 
\end{array}
\]
\[
\begin{array}{rcl}
P^{(2,3)} & = &
 x^{(2)}_1 \cdot x^{(3)}_1 + 	
 x^{(2)}_1 \cdot x^{(3)}_2 + 	
 x^{(2)}_1 \cdot x^{(3)}_3 + {}
 \\	
 &&
 x^{(2)}_2 \cdot x^{(3)}_1 + 	
 x^{(2)}_2 \cdot x^{(3)}_2 + 	
 x^{(2)}_2 \cdot x^{(3)}_3 + {}
 \\	
 &&
 x^{(2)}_3 \cdot x^{(3)}_1 + 	
 x^{(2)}_3 \cdot x^{(3)}_2 + 	
 x^{(2)}_3 \cdot x^{(3)}_3 
\end{array}
\]
\[
\begin{array}{rcl}
 P^{(2,4)} & = &
 x^{(2)}_1 \cdot x^{(4)}_1 + 	
 x^{(2)}_1 \cdot x^{(4)}_2 + 	
 x^{(2)}_1 \cdot x^{(4)}_3 + {}
 \\	
 &&
 x^{(2)}_2 \cdot x^{(4)}_1 + 	
 x^{(2)}_2 \cdot x^{(4)}_2 + 	
 x^{(2)}_2 \cdot x^{(4)}_3 + {}
 \\	
 &&
 x^{(2)}_3 \cdot x^{(4)}_1 + 	
 x^{(2)}_3 \cdot x^{(4)}_2 + 	
 x^{(2)}_3 \cdot x^{(4)}_3 \quad,
\end{array}
\]
\end{example}

\begin{claim}\label{Cla:UPPsize}
Let $\upp{P}{X}{n}{I}$ be a uniformly partitioned polynomial.
The following holds:
\[
\begin{array}{r@{\,\,}c@{\,\,}lcr@{\,\,}c@{\,\,}l}
\numM{P^{(a,b)}} & =& n^2
& \qquad & 
\numV{P^{(a,b)}} & =& 2 \cdot n
\\
\numM{P} & = & |I| \cdot n^2
& & 
\numV{P} & = & |X|\cdot n
\end{array}
\]
\end{claim}

\begin{example}\label{ex:uni-part-poly-claim}
Consider the polynomials $P$ and $P^{(1,2)}$ defined in \Cref{ex:uni-part-poly}.
The following holds:
\[
\begin{array}{lclcl}
\numM{P^{(a,b)}} & =& 3^2
\\
\numV{P^{(a,b)}} & =& 2 \cdot 3
& \text{as} & \Varsof{P^{(a,b)}} = 
\{ x^{(1)}_1, x^{(1)}_2, x^{(1)}_3,  x^{(2)}_1,  x^{(2)}_2,  x^{(2)}_3 \}
\\
\numM{P} & = & 4 \cdot 3^2
\\ 
\numV{P} & = & 4\cdot 3
& \text{as} & \Varsof{P} = 
\{ x^{(1)}_1, x^{(1)}_2, x^{(1)}_3,  x^{(2)}_1,  x^{(2)}_2,  x^{(2)}_3, 
\\
& & &  &
\hfill x^{(3)}_1, x^{(3)}_2, x^{(3)}_3,  x^{(4)}_1,  x^{(4)}_2,  x^{(4)}_3\}
\\
\end{array}
\]
\end{example}

\paragraph{Flat Abstractions}
We now define a special class of abstraction forests for uniformly partitioned polynomials. 
We refer to these abstraction forests as \emph{flat abstractions}.
Intuitively, the flat abstraction of a uniformly partitioned polynomial  $\upp{P}{X}{n}{I}$ partitions $P$'s variables into $|X|$ sets of equal size ($n$) according to the metavariable they pertain to.

\begin{definition}
An abstraction forest $\mathcal{T}$ is the \emph{flat abstraction} of a uniformly partitioned polynomial $\upp{P}{X}{n}{I}$ if
\[
\begin{array}{l}
\mathcal{T}=\{T_1,\ldots,T_{|X|}\}, \mbox{ where } T_i=(N_i,E_i) \mbox{ and }
\\
\quad  N_i=\{x^{(i)}\} \cup \{x^{(i)}_j \mid j \in \range{n}\}  
\\
\quad E_i=\{ (x^{(i)}, x^{(i)}_j) \mid j \in \range{n}\}\ . 
\end{array}
\]
\end{definition}

\begin{example}\label{ex:flat-abstraction}
\Cref{fig:flat-abstraction} depicts the flat abstraction of
the polynomial $P$  defined in \Cref{ex:uni-part-poly}.
\end{example}

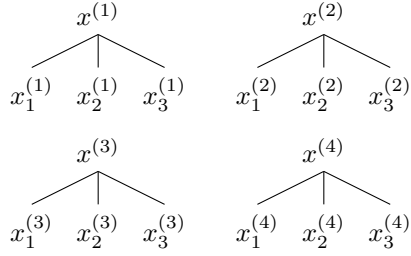
\begin{figure}[t]
\centering
\begin{tabular}{cc}
\begin{tikzpicture}
	\Tree [.$x^{(1)}$ $x^{(1)}_1$ $x^{(1)}_2$ $x^{(1)}_3$ 
	]
\end{tikzpicture}
& 
\begin{tikzpicture}
	\Tree [.$x^{(2)}$ $x^{(2)}_1$ $x^{(2)}_2$ $x^{(2)}_3$ 
	]
\end{tikzpicture}
\\
\begin{tikzpicture}
	\Tree [.$x^{(3)}$ $x^{(3)}_1$ $x^{(3)}_2$ $x^{(3)}_3$ 
	]
\end{tikzpicture}
& 
\begin{tikzpicture}
	\Tree [.$x^{(4)}$ $x^{(4)}_1$ $x^{(4)}_2$ $x^{(4)}_3$ 
	]
\end{tikzpicture}
\end{tabular}	
	\caption{A flat abstraction.}
	\label{fig:flat-abstraction}
\end{figure}

\begin{claim}\label{Cla:UPPAbsCompt}
The \emph{flat abstraction} of a uniformly partitioned polynomial $\upp{P}{X}{n}{I}$ is compatible with $P$.
\end{claim}

\begin{claim}\label{Cla:UPPAbssize}
Let $\mathcal{T}$ be the flat abstraction forest of a polynomial $\upp{P}{X}{n}{I}$, 
$S$ a valid variable set  of $\mathcal{T}$, and $Y\subseteq S$ the set comprised of the metavariables in $S$. 
It holds that 
$$
\begin{array}{rcl}
\numMsubV{P}{S}\! 
& = & 
\sum_{(i,j)\in I}
\begin{cases}
1 & x^{(i)} \in Y, x^{(j)}  \in Y \\
n^2 & x^{(i)} \not\in Y, x^{(j)} \not\in Y\\
n & \mathit{otherwise}
\end{cases}
\\[6pt]
\numVsubV{P}{S}\! 
& = & 
|Y| + (|X| - |Y|) \cdot n 
\end{array}
$$
\end{claim}

\begin{example}\label{ex:UPPAbssize}
Consider the polynomial $P$  defined in \Cref{ex:uni-part-poly}, its flat abstraction,
defined in \Cref{ex:flat-abstraction}, 
and the valid variable abstraction set    
$S=\{x^{(1)}, x^{(2)}_1, x^{(2)}_2, x^{(2)}_3, x^{(3)},  x^{(4)}_1,  x^{(4)}_2,  x^{(4)}_3 \}$. The set of metavariables in $S$ is
$Y=\{x^{(1)}, x^{(3)}\}$.
\[
 \subV{P}{S} = P_S^{(1,2)} + P_S^{(1,3)} + P_S^{(2,3)} + P_S^{(2,4)} ,  
\]
where
\[
\begin{array}{l}
\begin{array}{rcl}
P_S^{(1,2)} & = &
 3 \cdot x^{(1)} \cdot x^{(2)}_1 + 	
 3 \cdot x^{(1)} \cdot x^{(2)}_2 + 	
 3 \cdot x^{(1)} \cdot x^{(2)}_3 
\end{array}
\\[2ex]
\begin{array}{rcl}
P_S^{(1,3)} & = & 9 \cdot x^{(1)}  \cdot x^{(3)}
\end{array}
\\[2ex]
\begin{array}{rcl}
P^{(2,3)} & = &
 3 \cdot x^{(2)}_1 \cdot x^{(3)} + 	
 3 \cdot x^{(2)}_2 \cdot x^{(3)} + 	
 3 \cdot x^{(2)}_3 \cdot x^{(3)} 
\end{array}
\\[2ex]
\begin{array}{rcl}
 P^{(2,4)} & = &
 x^{(2)}_1 \cdot x^{(4)}_1 + 	
 x^{(2)}_1 \cdot x^{(4)}_2 + 	
 x^{(2)}_1 \cdot x^{(4)}_3 + {}
 \\	
 &&
 x^{(2)}_2 \cdot x^{(4)}_1 + 	
 x^{(2)}_2 \cdot x^{(4)}_2 + 	
 x^{(2)}_2 \cdot x^{(4)}_3 + {}
 \\	
 &&
 x^{(2)}_3 \cdot x^{(4)}_1 + 	
 x^{(2)}_3 \cdot x^{(4)}_2 + 	
 x^{(2)}_3 \cdot x^{(4)}_3 \quad.
\end{array}
\end{array}
\]
Note that 
\[
\Varsof{\subV{P}{S}} 
=
\{
x^{(1)}, x^{(3)}
\}
\cup 
\{
x^{(2)}_1, x^{(2)}_2, x^{(2)}_3, x^{(4)}_1,  x^{(4)}_2,  x^{(4)}_3
\}\,.
\]
\end{example}

\begin{claim}\label{Cla:NoZeroAbs}
Let $\mathcal{T}$ be the flat abstraction forest of a polynomial $\upp{P}{X}{n}{I}$, 
and $S$ a valid variable set  of $\mathcal{T}$.
It holds that $0 < \numMsubV{P}{S}$. 
\end{claim}
The claim holds for a rather mundane reason:
the coefficients of the monomials in $P$ are always positive. 
Thus, the abstraction process can only ``unify'' multiple monomials into one, but it never reaches a point when one monomial cancels another.

\begin{theorem}\label{the:multi}
Given a uniformly partitioned polynomial  $\upp{P}{X}{n}{I}$, 
a size bound $B\in\range[2]{(|I| \! \cdot\! n ) \!-\!1}$, and 
a granularity bound $K\in\range[2]{(|X| \! \cdot\! n ) \!-\!1}$,
determining if $P$ has a precise abstraction for $B$ and $K$ according to its flat abstraction forest is NP-hard. 
\end{theorem}

\paragraph{Reduction}
We prove \Cref{{the:multi}} via a reduction from the \emph{vertex cover} problem, a decision problem which we now define.
For notational simplicity, and without loss of generality, we omit certain easy cases from the vertex cover problem. Namely, we require that the input graph contains at least two node and one edge and 
forbid self loops.

\begin{definition}[Vertex cover]
Let $G=(V,E)$ be an undirected graph.
A set $V' \subseteq V$  is a \emph{vertex cover} of $G$ if for every edge $(x,y)\in E$ 
either $x \in V'$ or $y \in V'$.
\end{definition}

\begin{theorem}[\cite{Garey:1990}]\label{the:VC}
Given an undirected graph $G=(V,E)$,
where $V=\{v_1,\ldots,v_n\}$ for some $1<n$, 
$E \neq \emptyset$, and  
for any edge $(v_i,v_j)\in E$ it holds that $i \neq j$. 
Determining if $G$ has a vertex cover $V'$ such that $|V'|=k$, when
$k\in \range[2]{|V|\!-\!1}$, is NP-hard. 
\end{theorem}

 \begin{lemma}\label{lem:redVC}
Let  $G=(V,E)$ be an undirected graph which satisfies  the condition of 
\Cref{the:VC}.
$G$ has a vertex cover $V'\subset V$ of size $k\in\range[2]{|V|\!-\!1}$
if and only if 
the uniformly partitioned polynomial $\upp{P}{X}{|V|^3}{I}$, where 
$X=\{ x^{(i)} \mid n_i \in V\}$ and  $I=\{ (i,j) \mid (v_i,v_j)\in E \land i < j\}$, has a precise abstraction  
according to $P$'s flat abstraction forest 
for some $B\in\range[2]{|V|^5}$ and $K=(|V|-k) \cdot|V|^3 + k$.
\end{lemma}

\begin{proof}
Assume that there exists a valid variable set $S$ for $\mathcal{T}$ 
which  compresses $P$ into a polynomial $P'=\subV{P}{S}$ such  that  
$\numM{P'}=B$ and $\numV{P}=K$.
We show that the set of nodes $V'=\{ v_i \in V \mid x^{(i)} \in S\}$ corresponding to the metavariables in $S$ is a vertex cover of $G$ of size  $k$.

To see that $V'$ is a vertex cover,  
recall that   $P$ is uniformly partitioned. 
Hence, \Cref{Cla:UPPsize} ensures that each edge in $E$
corresponds to $|V|^6$ unique monomials  
and that $\numM{P}=|E|\cdot|V|^6$.
By assumption, $0 < |E|$.
Thus  $0 < \numM{P}$, and, by \Cref{Cla:NoZeroAbs}, $0 < \numM{P'}$.
By assumption, $1 < |V|$ and $B \leq |V|^5$. 
Hence, $B < |V|^6$.
By \Cref{Cla:UPPAbssize}, if there had been even a single  edge $(v_i,v_j)\in E$ such that  
$\{x^{(i)},x^{(j)}\}\cap V'=\emptyset$, it would mean  
that $|V|^6 \leq \numM{P'}$,
which would be a contradiction.

To see that $|V'|=k$, we use, again, \Cref{Cla:UPPAbssize},
and 
observe  that 
$\numV{P'} = (|V| - |V'|) \cdot |V|^3 + | V'|$.
Because $S$ is a precise abstraction,
$\numV{P'} = K$.
By assumption, $K = (|V| - k) \cdot |V|^3 + k$.
Which implies that $(|V|^3-1)\cdot(|V'|-k)=0$.
As $1 < |V|$, it must be the case that $|V|'-k$.

\smallskip
Assume $V'\subseteq V$ is a vertex cover of $G$ of size $k$.
Let \\
\mbox{$ 
 S=\{x^{(i)} \mid v_i \in V' \} \cup \{ x^{(i)}_j \in \Varsof{P} \mid v_i \not\in V'\}   
 $} 
be the set of meta\-variables annotating the roots of the abstraction trees in $\mathcal{T}$ corresponding to the nodes in the vertex cover,  
combined with the leaves of the other trees.
It is easy to see that $S$ is a  valid variable set for 
the flat abstraction forest of $P$.

As $V'$ is a vertex cover, \Cref{Cla:UPPAbssize} ensures that 
$\numMsubV{P}{S} \leq |E| \cdot |V|^3$. Clearly, $|E| \leq |V|^2$.
As $1 < |V|$, we get that\\ $\numMsubV{P}{S} \in \range[2]{|V|^5}$.
By construction, $|X|=|V|$ and $|S|=|V'|$.
By assumption, $|V'|=k$.
Hence, \Cref{Cla:UPPAbssize} ensures that\\  
$\numVsubV{P}{S}=(|V|-k)\cdot|V|^3 + k$, i.e.,
$\numVsubV{P}{S}=K$ as required.
\end{proof}

\begin{proof}[\textsc{Proof (\Cref{the:multi})}]
The proof follows directly from \Cref{the:VC} and \Cref{lem:redVC}:
The former states that the vertex cover problem for arbitrary $G=(V,E)$ and 
$k <|V|$ is NP-hard. 
The latter shows that we can solve the vertex cover problem by invoking the solution to our decision problem at most $|V|^5$ times. 
\end{proof}

%% file: varNumExp.tex
\section{Experiments: Additional Information}
\begin{figure}
	\centering
		\includegraphics[width = \linewidth]{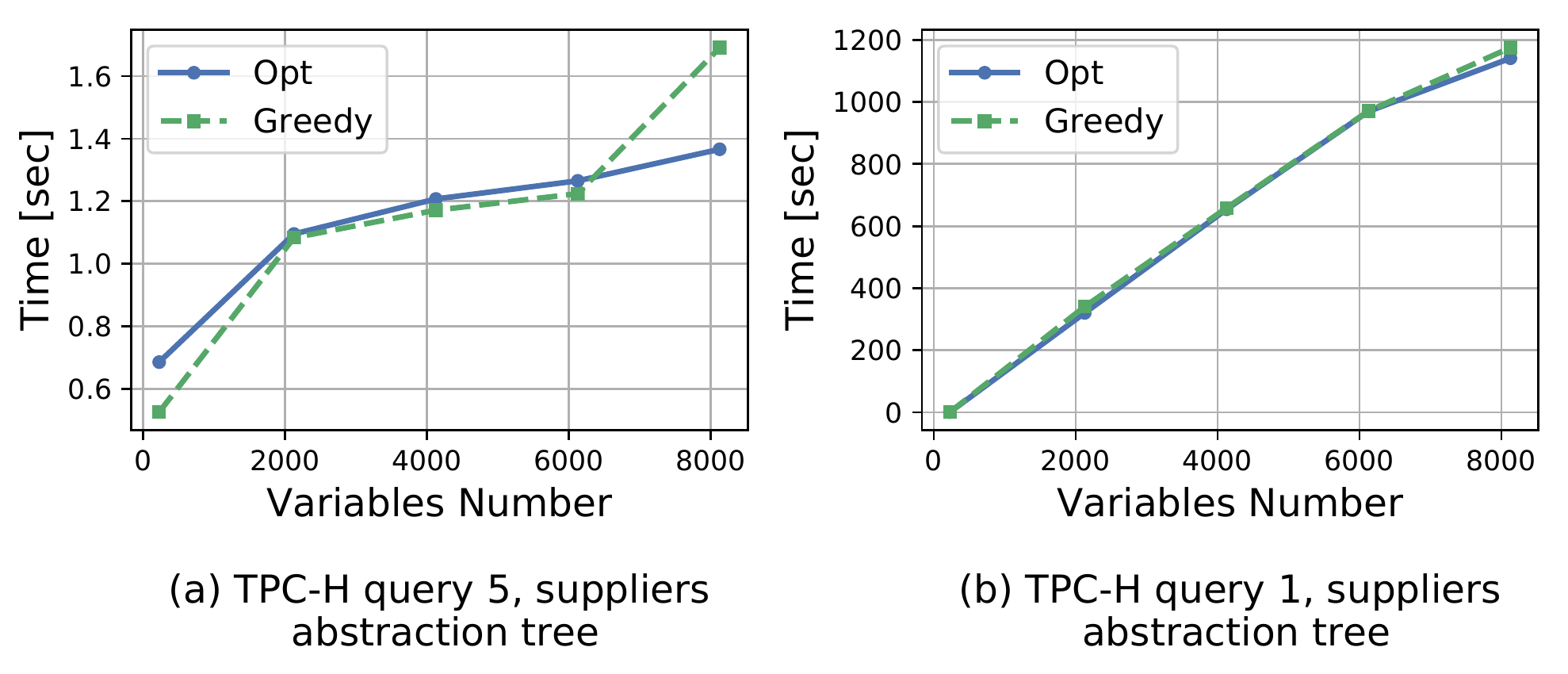}
		\caption{Compression Time as a Function of the Number of Variables}
			\label{fig:compVarNum}
\end{figure}

\mypara{Number of Variables}
To examine the affect of the number of variables on the algorithms' runtime, we used TPC-H 1Gb, and populated the database of our running example with 10,000 customers. The bound was set to 0.5 of the input polynomials size, and the number of variables varied up to 8000, of which 128 where the leaves of the abstraction tree.
Figure \ref{fig:compVarNum} presents the provenance compression time as a function of the number of variables in the input data for TPC-H queries 5 and 1. For those queries we observed moderate growth in the computation time, while the running time for query 10 and the running example query was roughly the same in all cases, and therefore the graphs are omitted. The difference stem from the affect of the variable number on the provenance size. Queries 1 and 5 generates relatively small number of polynomials, and thus, introducing new variables significantly increase the number of monomials per polynomial and the total provenance size, while query 10 and the running example query generate a large number of polynomials (993306 and 100,000 respectively), thus the affect of the additional variables on the total polynomials size is minor.

\begin{table}\small
	\begin{tabular}{ccccccc}
		\hline
		\multirow{2}{*}{Type} & \multirow{2}{*}{Nodes} & \multicolumn{4}{c}{Fan-out}        & \multicolumn{1}{c}{\multirow{2}{*}{VVS}} \\
		&                        & Root & Level 1 & Level 2 & Level 3 & \multicolumn{1}{l}{}                      \\ \hline \hline
		\multirow{6}{*}{1}    & 131                    & 2    & 64      & -       & -       &    5             \\
		& 133                    & 4    & 32      & -       & -       &          17                \\
		& 137                    & 8    & 16      & -       & -       &      257                               \\ 
		& 145                    & 16   & 8       & -       & -       &     65537  \\ 
		& 161                    & 32   & 2       & -       & -       &     4294967297                                  \\ 
		& 193                    & 64   & 2       & -       & -       &     1.84467E+19
		\\ \hline
		\multirow{5}{*}{2}    & 135 & 2    & 2      & 32       & -       &        26                                   \\
		&             139        & 2    & 4      & 16       & -       &         290                                  \\
		&              147       & 2    & 8      & 8       & -       &  66050                                         \\ 
		&               163      & 2   & 16       & 4       & -       &          4295098370                                 \\ 
		&              195       & 2   & 32       & 2       & -       &                1.84467E+19
		\\ \hline
		\multirow{4}{*}{3}    & 141 & 4    &   2   & 16       & -       &626                                           \\
		&                  149   & 4    & 4      & 8       & -       &   83522                                        \\
		&                   165  & 4    & 8      & 4       & -       &        4362470402                                   \\ 
		&                  197   & 4   & 16       & 2       & -       &                  1.84479E+19
		\\ \hline
		\multirow{3}{*}{4}    & 153 & 8    & 2      & 8       & -       &       390626                                    \\
		&                169     & 8    & 4      & 4       & -       &     6975757442                                      \\
		&                201     & 8    & 8      & 2       & -       &     1.90311E+19
		\\  \hline
		\multirow{4}{*}{5}    & 143 & 2    & 2      & 2       & 16       &       677                                    \\
		&                  151   & 2    & 2      & 4       & 8       &        84101                                   \\
		&                  167   & 2    & 2      & 8       & 4       &       4362602501                                    \\ 
		&                  199   & 2   & 2       & 16       & 2       &         1.84479E+19
		\\\hline
		\multirow{3}{*}{6}    & 155 & 2    & 4      & 2       & 16       &       391877                                    \\
		&              171       & 2    & 4      & 4       & 4       &        6975924485                                   \\
		&                203     & 2    & 4      & 2       & 8       &         1.90311E+19
		\\  \hline
		\multirow{3}{*}{7}    & 157 & 4    & 2      & 2       & 8       &     456977                                      \\
		&                173     & 4    & 2      & 4       & 4       &       7072810001                                    \\
		&                   205  & 4    & 2     & 8       & 2       &            1.90323E+19
		\\ \hline
	\end{tabular}
	\caption{Abstraction trees Types}
	\label{tbl:treesTypes}
\end{table}